\documentclass[10pt,final,journal]{IEEEtran}

\usepackage{cite}
\usepackage{url}
\usepackage{amsmath}
\usepackage{amssymb}
\usepackage{dsfont}
\usepackage{bm}
\usepackage{enumitem}
\usepackage[caption=false,font=footnotesize]{subfig}
\usepackage{siunitx}
\usepackage{algorithm}
\usepackage{graphicx}
\usepackage{epstopdf}
\usepackage{booktabs}
\usepackage{diagbox}
\usepackage{color}
\usepackage{soul}

\graphicspath{{figures/}}

\def\diag{\text{diag}}

\def\sgn{\text{sgn}}
\def\sat{\text{sat}}

\def\twon #1{\|#1\|}

\newcommand{\mathletter}[1]{%
	\expandafter\newcommand\csname b#1\endcsname{\mathbb #1}
	\expandafter\newcommand\csname c#1\endcsname{\mathcal #1}
	\expandafter\newcommand\csname f#1\endcsname{\mathfrak #1}
	\expandafter\newcommand\csname til#1\endcsname{\widetilde #1}
	\expandafter\newcommand\csname ha#1\endcsname{\widehat #1}
	\expandafter\newcommand\csname bf#1\endcsname{\bf #1}
}%
\def\mathletters#1{\mathlettersB #1,,}
\def\mathlettersB#1,{\ifx,#1,\else\mathletter #1\expandafter\mathlettersB\fi}
\mathletters{A,B,C,D,E,F,G,H,I,J,K,L,M,N,O,P,Q,R,S,T,U,V,W,X,Y,Z}

\def \qed {\hfill \vrule height6pt width 6pt depth 0pt}
\def\bee{\begin{equation}}
	\def\ene{\end{equation}}
\def\beq{\begin{eqnarray}}
	\def\enq{\end{eqnarray}}
\def\bmatri{\begin{bmatrix}}
\def\ematri{\end{bmatrix}}

\newtheorem{defi}{Definition}

\newtheorem{lemma}{Lemma}

\newtheorem{prop}{Proposition}

\newenvironment{proof}{\begin{IEEEproof}}{\end{IEEEproof}}

\title{Coordinate-free Circumnavigation of a Moving Target via a PD-like Controller}

\author{Fei~Dong, Keyou~You,~\IEEEmembership{Senior Member,~IEEE}, Lihua~Xie,~\IEEEmembership{Fellow,~IEEE}, Qinglei~Hu,~\IEEEmembership{Senior Member,~IEEE} 
	\thanks{*This work was supported in part by the National Natural Science Foundation of China under Grant 62033006 and Grant 61960206011, in part by the Beijing Natural Science Foundation under Grant JQ19017, and in part by a grant from the Guoqiang Institute, Tsinghua University. ({\em Corresponding author: Keyou You})}
	 \thanks{F. Dong and K. You are with the Department of Automation, Beijing National Research Center for Information Science and Technology, Tsinghua University, Beijing 100084, China. E-mail: dongf17@mails.tsinghua.edu.cn, youky@tsinghua.edu.cn.}%
	 \thanks{L. Xie is with the School of Electrical and Electronic Engineering, Nanyang Technological University, Singapore 639798, Singapore. E-mail: elhxie@ntu.edu.sg.}%
	 \thanks{Q. Hu is with the School of Automation Science and Electrical Engineering, Beihang University, Beijing 100191, China. E-mail: $\text{huql\_buaa}$@buaa.edu.cn.}%
}

\begin{document}

\maketitle

\begin{abstract}
This paper proposes a coordinate-free controller for a nonholonomic vehicle to circumnavigate a fully-actuated moving target by using \emph{range-only} measurements. If the range rate is available, our Proportional Derivative (PD)-like controller has a simple structure as the standard PD controller, except the design of an additive constant bias and a saturation function in the error feedback. We show that if the target is stationary,  the vehicle asymptotically encloses the target with a predefined radius at an exponential convergence rate, i.e., an exact circumnavigation pattern can be completed. For a moving target, the circumnavigation error converges to a small region whose size is shown proportional to the maneuverability of the target, e.g., the maximum linear speed and acceleration. Moreover, we design a second-order sliding mode (SOSM) filter to estimate the range rate and show that the SOSM filter can recover the range rate in a finite time. Finally, the effectiveness and advantages of our controller are validated via both numerical simulations and real experiments.

\end{abstract}


\section{Introduction}\label{sec1}

The target circumnavigation requires a mobile vehicle to enclose a target of interest at a stand-off distance to neutralize the target by restricting its movement \cite{wang2020stationary,kokolakis2021robust,olavo2018robust,oh2015coordinated,yoon2013circular,swartling2014collective}, which has been widely applied in both military and civilian applications for convoy protection or aerial surveillance purposes. The existing circumnavigation methods can be roughly categorized by the use of the state information of the vehicle or the target.

If the states (position, velocity, course, etc.) of both the vehicle and target are available, a Lyapunov guidance vector field (LGVF) method is proposed in \cite{Lawrence2003Lyapunov} and then extended in  \cite{Frew2008Coordinated,chen2013uav,Dong2019Flight}. Vector field methods are also proposed for the circular orbit tracking in \cite{Beard2014Fixed} and \cite{yao2021singularity}. Interestingly, the circumnavigation pattern can cover the moving path following (MPF) problem in \cite{Oliveira2016Moving}, which is resolved by designing a Lyapunov-based MPF control law and a path-generation algorithm.

For an uncooperative target, its state cannot be directly accessed by the tracking vehicle. In this case, the challenge lies in effectively estimating the target state via sensor measurements, such as ranges \cite{Shames2012Circumnavigation,swartling2014collective,cao2019relative}, bearings \cite{wang2020stationary,Deghat2014Localization}, or received signal strengths. For a stationary target, an adaptive localization algorithm is devised by using range-only measurements in \cite{Shames2012Circumnavigation} and a discrete-time observer is given in \cite{chai2014consensus} by using both the range and range rate  (the time derivative of range) measurements.  For a moving target, an adaptive motion estimator and a nonlinear filter are exploited in  \cite{Dobrokhodov2008Vision} and \cite{zhang2015Nonlinear}, respectively. Moreover, a Rao-Blackwellised particle filter is devised for a maneuvering target to simultaneously estimate its input and state in \cite{Dong2019Flight}. Note that the vehicle state is necessary to locate the target in the above mentioned works.

If neither the vehicle state nor the target state is available, e.g., the vehicle travels in complex underwater environments, a geometrical guidance law is designed in \cite{Cao2015UAV} by using a pair of a trigonometric function and an inverse trigonometric function, whose idea is to drive the vehicle towards a tangent point of an auxiliary circle. However, the control input is set as zero when the vehicle enters the circle, which may result in large overshoots. A biased proportional controller is proposed in \cite{zhang2020range} by using range rate measurements, which is also consistent with the bearing-only controller in \cite{park2016circling}. Moreover, a nonlinear PD controller is designed  in  \cite{Milutinovi2017Coordinate} for a state-space kinematic model which is composed of two continuous and one discrete state variables. Then, the control parameters depend on the maximum range of the controller operating space. Since the range-based controllers mentioned above are only concerned with the circumnavigation problem of a \emph{stationary} target,  how to extend to the  moving target  is unclear. 

There is no doubt that circumnavigating a moving target is more practical and significant. To this end, a sliding mode controller (SMC) is proposed in \cite{Matveev2011Range,matveev2017tight}. To eliminate the chattering phenomenon, they model the dynamics of actuator as the simplest form of the first-order linear differential equation. Yet, their approach cannot achieve zero steady-state error even for a stationary target, and requires the vehicle to start far away from the target. Anderson et al. \cite{Anderson2014A} devise a stochastic approach by further using relative angles. Shames {\em et al.} \cite{Shames2012Circumnavigation} show that the upper bound of the circumnavigation error is proportional to the maximum linear speed of the moving target,  which however uses the explicit position information of the vehicle. In sharp contrast, we can achieve the same results by using the {range-only} measurements. 
 
Specifically, we propose a PD-like controller to ensure that the vehicle can circumnavigate a moving target with range-only measurements, the idea of which is inspired from our previous work \cite{dong2019ICCA}. Indeed, the controller in \cite{dong2019ICCA} also has a PD-like form. Moreover, it investigates the case of a stationary target with both range and range rate measurements by using time-varying gains. Given that the target is stationary and both range and range rate measurements are known, our controller is shown to achieve the exact circumnavigating task at an exponential convergence rate. This implies that the closed-loop system is robust against small perturbations, which is essential to explicitly derive the upper bound of the circumnavigation error for the case of a time-varying target. Such a result is also consistent with \cite{Shames2012Circumnavigation} though the latter requires the position information of vehicle.  Moreover, the error bound can be further reduced by selecting proper control parameters, which are independent of the initial state. 

In some scenarios, the range rate may not be accessible to the vehicle. Since small noises may result in large or unbounded estimation errors, it is not effective to simply calculate the range rate by regular differentiating methods. To address it, a first-order filter and a washout filter are utilized in \cite{Target2018Guler} and \cite{Lin20163}, respectively, which unfortunately lack a rigorous justification. A second-order sliding mode (SOSM) filter is proposed in  \cite{Moreno2012Strict} and is adopted for the circumnavigating problem in \cite{Cao2015UAV} and \cite{zhang2020range} for a \emph{stationary} target. Similarly, we revise our controller into a range-only form by designing an SOSM filter, the estimation error of which converges to zero in a finite time if the initial distance to the target is sufficiently large and both the speed and acceleration of the target are bounded.

In a nutshell, the contributions of this paper are summarized as follows: 
\begin{itemize}
	\item[(a)] A PD-like controller is proposed to solve the circumnavigation problem by only using the range-based measurements, and is shown to complete the exact circumnavigation task at an exponential convergence rate if the target is stationary.
	\item[(b)] For a moving target, the steady-state circumnavigation error can be arbitrarily  reduced by increasing the P gain if there is no controller limit.
	\item[(c)] An SOSM filter is further designed to recover the range rate in a finite time by using range-only measurements for a moving target with bounded velocity and acceleration. 
\end{itemize}

The rest of this paper is organized as follows. In Section \ref{sec2}, the target circumnavigation problem is described in detail. In Section \ref{sec_controller}, we propose the PD-like controller to regulate the nonholonomic vehicle. If the explicit range rate is known, we prove the exponential convergence and derive the circumnavigation error bound for the moving target in Section \ref{sec3}.  We extend to the case without explicit range rate in Section \ref{sec5}. Both simulations and experiments are performed in Section \ref{secsim}. Some concluding remarks are drawn in Section \ref{sec6}.

\section{Problem Formulation}\label{sec2}

\begin{figure}[t!]
	\centering{\includegraphics[width=0.8\linewidth]{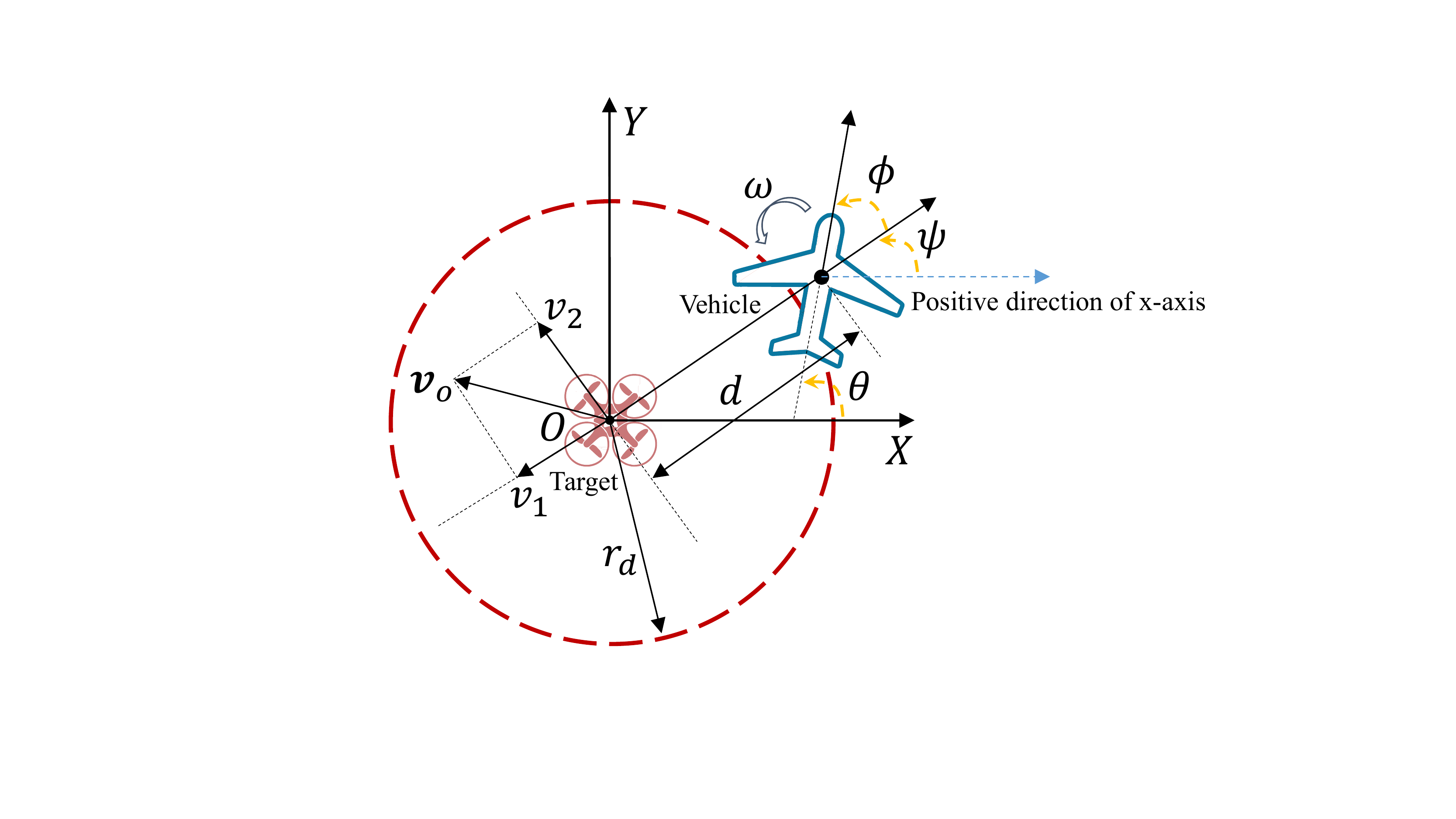}}
	\caption{Circumnavigation of a moving target.}
	\label{fig2}
\end{figure}
In Fig.~\ref{fig2}, we aim to track a moving target with double-integrator kinematics on a horizontal plane
\begin{align} \label{eqtar}
	\dot {\bm p}_o(t) = \bm v_o(t) ,~ \dot {\bm v}_o(t) = \bm a_o(t) ,
\end{align} 
by a nonholonomic vehicle  
\begin{align} \label{eqrob}
	\dot {\bm p}(t) = v\bmatri \cos \theta(t) \\  \sin \theta(t) \ematri, ~\dot \theta(t) = \omega(t) ,
\end{align}
where $\bm p_o(t)$, $\bm v_o(t)$, ${\bm a_o(t)}\in \mathbb{R}^2$ denote the position, linear velocity, and acceleration of the target, and $\bm p(t)\in \mathbb{R}^2,~\theta(t)$,~  $\omega(t)$, $v$ represent the position, heading course, adjustable angular speed, and constant linear speed of the vehicle, respectively.  If the target and the vehicle travel with different altitudes, e.g., an unmanned aerial vehicle (UAV) circles over a ground moving target, then Fig.~\ref{fig2} denotes the projection of the 3D trajectory of the UAV to the horizontal plane. In this work, only the (horizontal) range measurement from the vehicle to the target is available, i.e.,
\begin{align}\label{eqrange}
	d(t) = \twon{\bm p(t) - \bm p_o(t)}_2.
\end{align}
Neither the target position $\bm p_o(t)$ nor the vehicle position $\bm p(t)$ is known, which requires the controller to be coordinate-free. A notable example is that the vehicle travels in GPS-denied environments and the target is an uncooperative intruder.

Our objective is to design a coordinate-free controller\footnote{Coordinate-free  refers to that the controller is designed without  the position information of the vehicle.} via range-\emph{only} measurements $d(t)$ to drive the vehicle \eqref{eqrob} to circumnavigate the target \eqref{eqtar} with a predefined radius $r_d$. Mathematically, 
\begin{itemize}
	\item[(i)] if $\bm p_o(t) \equiv \bm p_o$ in \eqref{eqtar} is constant, it requires that
	\begin{align} \label{eqobj1}
		\lim_{t\rightarrow \infty} |d(t)-r_d |=\lim_{t\rightarrow \infty} |\dot d(t) |=0,
	\end{align}
	\item[(ii)] if both $\Vert \bm v_o(t) \Vert_2 \le \bar v_o<v$ and $\Vert \bm a_o(t) \Vert_2 \le \bar a_o$ in \eqref{eqtar} for all $t\ge t_0$, it requires that
	\begin{align} \label{eqobj2}
		\limsup_{t\rightarrow \infty} |d(t)-r_d |\le \epsilon,
	\end{align}
	where $\epsilon>0$ is a small constant that explicitly depends on $\bar v_o$ and $\bar a_o$.       
\end{itemize}
That is, the trajectory of the vehicle in \eqref{eqobj1} asymptotically forms an exact circle with the unknown target position ${\bm p}_o$ and  $r_d$ as its center and radius, respectively. In \eqref{eqobj2}, the trajectory is close to a circle with ${\bm p}_o(t)$ and $r_d$ as its moving center and radius. 

In \cite{Shames2012Circumnavigation},  the two objectives in \eqref{eqobj1} and \eqref{eqobj2} have been achieved by a \emph{single-integrator} vehicle but further using its exact position information $\bm p(t)$. By contrast, our coordinate-free controller does not need the position $\bm p(t)$ of the \emph{nonholonomic} vehicle.

\section{Controller Design}\label{sec_controller}
In this section, we propose two controllers. The first is the range-based controller by using both the range measurement $d(t)$ and its rate $\dot d(t)$. It has a PD-like form with a bias to eliminate the steady-state circumnavigation error. A similar idea has been presented in our preliminary work \cite{dong2019ICCA}, which however only investigates the case of a stationary target. The second is the range-only controller by further designing an SOSM filter to recover the range rate if both the linear speed and acceleration of the target are bounded. 

\subsection{The PD-like controller with explicit range rates}
To solve the circumnavigation problem,  define a relative tracking error 
\bee\label{errordef}
e(t)=\frac{d(t)-r_d}{r_d},
\ene
and a saturation function 
\begin{align*}
	\text{sat}(z)=
	\begin{cases}
		z,& ~\text{if}~ {|z| < 1},\\
		\sgn (z),& ~\text{if}~ {|z| \ge 1},
	\end{cases}
\end{align*}
where $\sgn(\cdot)$ is the standard sign function.

If the range rate $\dot d(t)$ is explicitly known, we propose the following PD-like \emph{range-based} controller 
\begin{align} \label{eq222}
	\omega(t)  = \omega_c + c_1 \dot e(t) + c_2 \sat \left( e(t)\right),
\end{align}
where $c_i$, $i\in\{1,2\}$ is a positive parameter to be designed and $\omega_c={v}/{r_d}$ is a bias to eliminate the steady-state circumnavigation error. Specifically, if $d(t) =r_d$ and $\dot d(t)=0$ at some time $t$, then $\omega(t)  = \omega_c$ is the desired angular speed of the vehicle, and if the target is stationary, the vehicle will keep this angular speed afterwards. 

The major difference from the standard PD controller lies in the use of the saturation function to ensure that the control parameters can be selected \emph{independent} of the initial state of the circumnavigation system. That is, if we remove the saturation function, we also need to restrict the initial state of $e(t)$ for a fixed $c_i$ since $|\dot e(t)|$ is bounded. Otherwise, the circumnavigation task may fail even for a stationary target. As the system is inherently nonlinear, one cannot expect to use a linear controller to \emph{globally} complete the circumnavigation task.  From this perspective, our controller in \eqref{eq222} is the ``simplest" one. 

For the case of a stationary target, we show in Proposition \ref{prop2} that the PD-like range-based controller in \eqref{eq222} can even achieve an exponential convergence with a fixed set of parameters for any initial condition. In comparison, the sliding mode approach in  \cite{Matveev2011Range}  cannot achieve exact circumnavigation, i.e., $e(t)$ cannot exactly converge to zero. The geometrical method in  \cite{Cao2015UAV} may result in large overshoots since there is no control input when the vehicle enters an auxiliary circle. Moreover, the control parameters in \cite{Milutinovi2017Coordinate} are determined by the maximum range of the controller operating space. Importantly, both controllers in  \cite{Cao2015UAV} and \cite{Milutinovi2017Coordinate} are only concerned with a stationary target, and it is confirmed by Fig.~\ref{fig12} in Section \ref{secsim} that their controllers cannot be adopted for the case of a moving target. Furthermore, the backstepping controller in \cite{dong2019Target} aims to steer the vehicle to follow a \emph{smooth} reference command from a stationary target. 

Since the PD-like controller in \eqref{eq222} only contains the range-based measurements from the vehicle to the target, it is particularly useful in GPS-denied environments and substantially different from \cite{Dobrokhodov2008Vision,Dong2019Flight,Oliveira2016Moving,zhang2015Nonlinear,Deghat2014Localization} as they further require the position information $\bm p(t)$. 

\subsection{The PD-like controller without explicit range rates}
If the range rate $\dot{d}(t)$ is unavailable, we adopt an SOSM filter \cite{Moreno2012Strict} to estimate it, i.e., 
\begin{equation} \label{eqfilter}
\left\{
\begin{split} 	
		\dot {\alpha}_1(t) &=  k_1 |d(t) -\alpha_1(t)|^{1/2} \sgn\left(d(t) -\alpha_1(t) \right)\\
		&~~~~+k_2(d(t) -\alpha_1(t) ) +  \alpha_2(t) \\
		\dot {\alpha}_2(t) & =k_3\sgn\left(d(t) -\alpha_1(t) \right)+k_4\left(d(t) -\alpha_1(t) \right)  
	\end{split},
\right.
\end{equation}
where $k_i$, $i\in\{1,2,3,4\}$ is a positive filter parameter to be designed. If both the linear speed and acceleration of the target are bounded, we show that there is a finite $T$ such that
$
d(t)-\alpha_1(t) = \dot d(t)-\alpha_2(t)=0, ~\forall t \ge t_0+T.
$

Thus, we can directly replace $\dot e(t)$ in \eqref{eq222} by $\alpha_2(t)/r_d$ and obtain the following PD-like \emph{range-only} controller 
\begin{align} \label{eq2}
	\omega(t)  = \omega_c  +  c_1/r_d \cdot \alpha_2(t) + c_2 \sat \left( e(t)\right).
\end{align}

\section {Moving Target Circumnavigation under the Range-based Controller} \label{sec3}
If the target is stationary, the PD-like range-based controller in \eqref{eq222} can achieve an exponential convergence with a fixed set of parameters for any initial condition. Otherwise, the upper bound of the circumnavigation error is explicitly shown to be proportional to the maximum linear speed and acceleration of the moving target. 

\subsection {Stationary target circumnavigation} \label{subsec3}
For a stationary target, i.e., $\bar v_o = 0$, let $\bm p_o$ be the \emph{unknown} position of the target.  Consider a polar frame centered at the target, we convert the kinematics in \eqref{eqrob} from the Cartesian coordinates into the following form
\begin{equation} \label{eq5}
	\begin{split}
		\dot d(t) &= v \cos \phi(t),\\
		\dot\phi(t) & = \omega(t) - \frac{v}{d(t)} \sin \phi(t) ,
	\end{split}
\end{equation}  
where the angle $\phi(t) \in (-\pi,\pi]$ is formed by the direction that the target points to the vehicle and the heading direction of the vehicle. See Fig.~\ref{fig2} for an illustration. By convention, the counter-clockwise direction is set to be positive. From Fig.~\ref{fig2}, we obtain that $\phi(t) = \theta(t)  - \psi(t)$,
where $\psi(t)$ is subtended by the direction from the target to the vehicle and the positive direction of $x$-axis.

Note that, $\phi(t)$ is not defined when $d(t)=0$. In light of Fig.~\ref{fig2}, the case $d(t)=0$ is a special one in which the vehicle goes directly through the target. Thus, we follow the definition in  \cite{Milutinovi2017Coordinate}. 

\begin{defi} \label{defi}	
	Suppose that there is a time instant $t_*>t_0$ such that $d(t_*)=0$, then the angle $\phi(t)$ just before hitting the target is $\phi(t_*^{-})=\pi$ and just after leaving the target is $\phi(t_*^{+})=0$.
\end{defi}

By Fig.~\ref{fig2},  one can easily observe that $[d(t),\phi(t)]'=[r_d, \pi/2]'$ is the desired state to achieve the objective of circumnavigation in the counter-clockwise direction, which is also an equilibrium of \eqref{eq5}. 

Now, we show that the closed-loop system in \eqref{eq5} with the range-based controller in (\ref{eq222}) is exponentially stable. 

\begin{prop} \label{prop2} Consider the circumnavigation system in (\ref{eq5}) under the PD-like range-based controller in (\ref{eq222}). Let $\bm x(t)=[d(t),\phi(t)]'$ and $\bm x_e =[r_d,\pi/2]'$.  If the control parameters are selected to satisfy that
	\begin{align} \label{eqc}
		{(c_1-1) \omega_c}> {c_2} ,
	\end{align}
	there exists a finite time instant $t_1\ge t_0$ such that 
	\begin{align*}
		\twon{\bm x(t)-\bm x_e} \le C\twon{\bm x(t_1)-\bm x_e} \exp\left(-\rho (t- t_1)\right), ~\forall t >  t_1,
	\end{align*}
	where $\rho$ and $C$ are two positive constants.  
\end{prop} 
\begin{proof}
	See Appendix \ref{suba1}.
\end{proof}	

It is clear that the convergence rate to the equilibrium $\bm x_e$ is exponentially fast. Thus, small perturbations will not result in large deviations from the equilibrium  \cite[Chapter 9.2]{Khalil2002Nonlinear}. Interestingly, the selection of control parameters $c_i$, $i\in \{1,2\}$ is independent of $d(t_0)$ in light of \eqref{eqc}, which is in sharp contrast to \cite{Milutinovi2017Coordinate}.

\subsection{Moving target circumnavigation} \label{subsec4}
For a moving target, we decompose its forward velocity $\bm v_o(t)$ into $v_1(t)$ and $v_2(t)$, which denote the radial and tangential velocities of the target relative to the vehicle, respectively. See Fig.~\ref{fig2} for an illustration. Then the circumnavigation kinematics is given by 
\begin{equation} \label{eq19}
	\begin{split}
		\dot d(t) &= v \cos\phi(t) - v_1(t),\\
		\dot\phi(t) &= \omega(t) - \frac{v}{d(t)} \sin \phi(t) + \frac{v_2(t)}{d(t)}.
	\end{split}
\end{equation}

Now, we show that the circumnavigation error of the closed-loop system in \eqref{eq19} under the range-based controller in (\ref{eq222}) is bounded by a constant, which is proportional to the maximum linear speed and acceleration of the target.  Moreover, we can reduce the upper bound of the  circumnavigation error by properly increasing $c_i$, $i\in\{1, 2\}$.

\begin{prop} \label{prop_moving}
	Consider the target circumnavigation system in (\ref{eq19}) under the range-based controller in (\ref{eq222}). If $\Vert \bm v_o(t)\Vert_2 \le \bar v_o$, $\Vert \bm a_o(t)\Vert_2 \le \bar a_o$, and the control parameters are selected to satisfy that 
	\begin{gather}
		(c_1-1)\omega_c > c_2 + (c_1+1) \omega_o,\label{condition1} \\
		c_2 > \max\left\{(c_1+1)\omega_o,~ 2\omega_c+4\omega_o\right\} \label{condition2},
	\end{gather}
	where $\omega_o = \bar v_o/r_d$, then there is a positive $\epsilon>0$ of the form 
	\begin{equation}\label{kappa}
		\epsilon=\mathcal{O}\left(\frac{v+\bar{v}_o+\bar{a}_o}{c_2}+\frac{v+\bar{v}_o}{c_1}\right)
	\end{equation}
	and a positive constant\footnote{An explicit form of $T_1$ is given in the proof of Lemma \ref{lemma_moving}.} $T_1=T_1(c_1,c_2,\omega_c )$ such that $$
	\limsup_{t\rightarrow \infty} |d(t)-r_d |\le \epsilon$$  for all ${d(t_0)> 2 r_d+(v+\bar v_o) T_1}$. 
\end{prop} 

\begin{proof}By \eqref{condition1}, we obtain that
	$c_1\omega_c>(c_1-1)\omega_c > c_2 + (c_1+1) \omega_o > c_2 + c_1 \omega_o$
			which jointly with $\omega_c=v/r_d$ and $\omega_o=\bar v_o/r_d$ implies $v>\bar v_o +c_2/c_1\cdot r_d$. Let $q_1 = {c_2/c_1\cdot r_d}$.  Then, there exists a  $v_*\in(0,v-\bar v_o -q_1)$ and $\phi_*$ such that
		\begin{align} \label{sinfun}
		\sin \phi_*=\left(1-((v_*+\bar v_o+q_1)/v)^2\right)^{1/2}.
		\end{align}

{\bf Step 1}: we show that there exists a finite time instant $t_1\ge t_0$ such that $\sin\phi(t)\ge\sin \phi_*>0 $ for all $t\ge t_1$.
		
	Let \begin{equation} 
	\label{ztdef}	z(t)=\dot e(t)+c_2/c_1\cdot \sat(e(t)).\end{equation}	
	By \eqref{eq19} and $z(t)=0$, it yields that 
	\begin{align*}
	\phi(t) = \arccos\left(\frac{v_1(t)-q_1\sat(e(t))}{v}\right),
	\end{align*}
	where $q_1$ is defined in \eqref{sinfun}. Then, inserting \eqref{eq222} into \eqref{eq19} leads to that
	\begin{equation} \label{eq63}
	\dot\phi(t)={c_1}z(t)+{v}/{r_d}  - {v}/{d(t)}\cdot \sin \phi(t) + {v_2(t)}/{d(t)}.
	\end{equation}
	
	(a) If $\phi(t_0) \in [\arccos( {(q_1+\bar v_o)/v}),\pi-\arccos((-v_*-\bar v_o-q_1)/v)]$, then $\phi(t)$ remains in this interval for all $t\ge t_0$.  
	When $\phi(t) = \arccos( (q_1+\bar v_o)/v)$, i.e., $z(t)/r_d=q_1+ q_1 \sat(e(t))+ \bar v_o -v_1(t)$, it follows from \eqref{condition2} and \eqref{eq63} that
	\begin{align*} 
	\dot\phi(t) 
	>&~ c_1/r_d\cdot (q_1/2) -\omega_c-2\omega_o>0.
	\end{align*}
	Similarly, $\phi(t) = \pi-\arccos((-v_*-\bar v_o-q_1)/v)$ leads to that
	$
	\dot\phi(t) 
	<- c_1/r_d \cdot v_* +\omega_c   + \omega_o<0.
	$ Since $\phi(t)$ is continuous in $t$, the result follows. 
	
	(b) If $\phi(t_0) \notin [\arccos( (q_1+\bar v_o)/v),\pi-\arccos((-v_*-\bar v_o-q_1)/v)]$, we show in Lemma \ref{lemma_moving} of Appendix \ref{suba2} that there is a finite $\delta>0$ such that $\phi(t_0+\delta) \in [\arccos( (q_1+\bar v_o)/v),\pi-\arccos((-v_*-\bar v_o-q_1)/v)]$.   
	
	Thus, there exists a finite time instant $t_1\ge t_0$ such that $\phi(t) \in [\arccos( {(q_1+\bar v_o)/v}),\pi-\arccos((-v_*-\bar v_o-q_1)/v)]$ for all $t\ge t_1$, i.e., $\sin\phi(t) \ge \sin \phi_*$.

{\bf Step 2}: we show the uniform boundedness of $z(t)$.  By {\bf Step 1}, there is no loss of generality to assume that $\sin\phi(t)\ge\sin \phi_*>0 $ for all $t\ge t_0$. 

Consider the following Lyapunov function candidate 
	\begin{align} \label{eqvz}
		V_z(z) = \frac{1}{2} z^2(t).
	\end{align}
	
	If $d(t)\ge 2r_d$, it follows from \eqref{eq222} and \eqref{eq19} that
	\begin{align*}
		\dot z(t)
		&= -\left( \omega_c - \frac{v}{d(t)} \sin \phi(t) + \frac{v_2(t)}{d(t)}\right) \omega_c\sin \phi(t)- \frac{\dot v_1(t)}{r_d}\\
		&~~~~  -{c_1 \omega_c}z(t) \sin \phi(t).
	\end{align*}
	Then, the time derivative of $V_z(z)$ leads to that
	\begin{align}\label{eqvz1}
		\dot V_z(z)
		&=- c_1\omega_c  z^2(t)\sin \phi(t) -z(t)\times \nonumber\\
		&~~~~ \left(\left( \omega_c - \frac{v\sin \phi(t)}{d(t)}  + \frac{v_2(t)}{d(t)}\right) \omega_c\sin \phi(t) + \frac{\dot v_1(t)}{r_d}\right)\nonumber\\
		&<-{c_1 \omega_c} z^2(t) \sin \phi(t)+|z(t)| \left(\omega_c \left( \omega_c+ \omega_o\right) + \frac{\bar a_o}{r_d}\right).
	\end{align}
	
	If $d(t)\in (r_d/2, 2r_d)$,\footnote{Note that the initial condition $d(t_0)>2r_d+(v+\bar{v}_o)T_1$ excludes the case $d(t)\le r_d/2$.} it similarly holds that 
	\begin{align}\label{eqvz2}
		\dot V_z(z) 
		&<|z(t)|\left(\omega_c \left( \omega_c +2\omega_o \right) + \frac{\bar a_o}{r_d}+ {c_2}\frac{v+\bar v_o}{c_1}\right)\nonumber \\
		&~~~~-{c_1 \omega_c} z^2(t)\sin \phi(t) . 
	\end{align}
	
	Overall, if $d(t)>r_d/2$, it follows from \eqref{eqvz1} and \eqref{eqvz2} that  
	\begin{align*}
		\dot V_z(z)<&~ |z(t)| \left(\omega_c \left( \omega_c+ 2\omega_o\right) + {\bar a_o}/{r_d}+ {c_2}(v+\bar v_o)/{c_1}\right)\\
		& -{c_1 \omega_c}z^2(t) \sin \phi_* .
	\end{align*}
	Thus, $\dot V_z(z)<0$ holds for all 
	\begin{align*}
		|z(t)|\ge \frac{\omega_c \left( \omega_c+ 2\omega_o\right) + \bar a_o/r_d+{c_2}(v+\bar v_o)/{c_1}}{{c_1\cdot \omega_c} \sin \phi_*}\triangleq\frac{\epsilon_1}{c_1},
	\end{align*}
	which means that $|z(t)|$ will be bounded by $\epsilon_1/c_1$. 
	
	{\bf Step 3}: we show the uniform boundedness of $e(t)$.  

By {\bf Step 2}, \eqref{ztdef} and Lemma 6.2 of \cite{Matveev2011Range}, it yields that 
	\begin{align} \label{kappa1}
		\limsup_{t\rightarrow \infty} |d(t)-r_d| \le {\epsilon_1 r_d}/{c_2}\triangleq \epsilon , 
	\end{align}
	where
	\begin{align*}
	 \epsilon
	=\frac{ v+ 2\bar v_o}{{c_2} \sin \phi_*} + \frac{\bar a_o}{c_2\omega_c \sin \phi_*} +\frac{r_d(v+\bar v_o)}{c_1\omega_c \sin \phi_*}.
	\end{align*}	
	Since $v_*$ in \eqref{sinfun} depends only on the control parameters in the form of $c_2/c_1$, then \eqref{kappa} follows from \eqref{kappa1}. 
\end{proof}

By \eqref{kappa},  the steady-state circumnavigation error $\epsilon$ is proportional to the maneuverability of the target, and can be made small by increasing the control parameters $c_1$ and $c_2$. If the target is stationary, i.e. $\bar{v}_o=\bar{a}_o=0$,  then $c_1$ and $c_2$ can be selected properly large, in which case the steady-state circumnavigation error is close to zero.   

	\subsection{The input saturation}
	If the vehicle in \eqref{eqrob} has a controller limit, i.e. $|\omega(t)|\le \bar \omega$, the limit $\bar \omega$ should be large enough to maintain the circumnavigation pattern, which is quantified below. 
	
	\begin{lemma}\label{lemmasat}
		The vehicle in \eqref{eqrob} can maintain the distance $r_d$ from the target \eqref{eqtar} if and only if 
		\begin{equation} \label{contrlimit}
		v \bar \omega  \ge  (v+\bar v_o)^2/r_d + \bar a_o.   
		\end{equation}
	\end{lemma} 	
	\begin{proof}
		The circumnavigation pattern is maintainable if and only if there exists a finite time instant $t_1\ge t_0$ such that $e(t)=\dot e(t)=0$ for all $t \ge t_1$. By \eqref{errordef}, \eqref{eq19} and the proof of \cite[Proposition 3.1]{Matveev2011Range},  the inequality in \eqref{contrlimit} is not difficult to establish. The details are omitted due to  space limitation.
	\end{proof}
	
	Thus,  the maximal feasible acceleration of the vehicle should be greater than that of the target plus the maximal radial acceleration, which is resulted from the rotation of the vehicle relative to the target at the radius $r_d$. In this regard, it is reasonable to assume that the controller limit $\bar \omega$ always satisfies the strict inequality in \eqref{contrlimit}. Then,  we revise the range-based controller in  \eqref{eq222} as follows
	\begin{align}\label{eq_sat}
	\omega_s(t)= \bar \omega\cdot \sat\left( \frac{1}{\bar \omega} \left( {\omega_c} + c_1 \dot e(t) + {c_2}\sat(e(t)) \right) \right). 
	\end{align}

	Since
	$\max_{t\ge t_0} |\omega(t)| = \omega_c + c_1(v+\bar v_o) + c_2$, the controller in \eqref{eq_sat} will never saturate if $c_1$ and $c_2$ are selected to satisfy $\omega_c + c_1(v+\bar v_o) + c_2<  \bar \omega$.  Jointly with \eqref{kappa},  the tracking error is essentially bounded by $\mathcal{O}(1/\bar \omega)$.

\section{Moving Target Circumnavigation under the Range-only Controller} \label{sec5}
If the range rate $\dot d(t)$ is not explicitly available, we adopt the SOSM filter in \eqref{eqfilter} to obtain the range-only controller \eqref{eq2}. A similar idea can also be found in \cite{Cao2015UAV}, which however only focuses on the \emph{stationary} target circumnavigation. Note that a first-order filter and a washout filter are adopted in \cite{Target2018Guler} and \cite{Lin20163}, respectively.  

\begin{prop} \label{prop4} Consider the circumnavigation system in (\ref{eq19}) under the PD-like range-only controller in (\ref{eq2}).  If $\Vert \bm v_o(t) \Vert_2 \le \bar v_o$, $ \Vert \bm a_o(t) \Vert_2\le \bar a_o$,  the parameters of the filter \eqref{eqfilter} and controller (\ref{eq2}) satisfy that
	\begin{gather*} 
		k_1 > 2\sigma_2, ~k_2 >\sigma_2^2 + 2\sigma_2,\\
		k_3>\max\left\{0,~ (k_1 +1)\sigma_1/k_1 -k_1^2/2,~\sigma_1-2k_1^2 -k_1^2/(2k_2)\right\},\\
		k_4 > \max\left\{0, ~k_2/2-k_2^2,~k_2^2(2k_1 + 5\sigma_1)/(k_1-2\sigma_2)\right\}\\
		(c_1-1)\omega_c > c_2 + (c_1+1) \omega_o, \\
		c_2 > \max\left\{(c_1+1)\omega_o,~ 2\omega_c+4\omega_o\right\},
	\end{gather*}
	where $\sigma_1= 2 \omega_c v + c_1 \omega_c v +c_2 v + \omega_c \bar v_o + \bar a_o$ and $\sigma_2=c_1 \omega_c$, and ${d(t_0)> 2r_d  + (v+\bar v_o)(T_1+T_2)}$ with a positive constant\footnote{The explicit form of $T_2$ is given in the proof of Proposition \ref{prop4}.} $T_2$, then it holds that
	\begin{align*}
		\alpha_1(t)=d (t)~\text{and}~ \alpha_2(t)=\dot d(t),~ \forall t > t_0+T_2.
	\end{align*}
	
	Moreover, $ \limsup_{t\rightarrow \infty} |d(t)-r_d |\le \epsilon $ where $\epsilon$ and $T_1$ are given in Proposition \ref{prop_moving}.
\end{prop}

\begin{proof}
	In light of Proposition \ref{prop_moving}, we complete the proof by showing that $\alpha_1(t)=d (t)~\text{and}~ \alpha_2(t)=\dot d(t)$  
	for any $t\ge t_0+T_2$, where $T_2$ is finite. Then, the circumnavigation system \eqref{eq19} exactly works as the case of using the explicit range rate $\dot d(t)$ after $t_0+T_2$. 

	To this end, we define the estimation error as $$\xi_1(t) = d(t)-\alpha_1(t)~\text{and}~\xi_2(t)=\dot d (t)-\alpha_2(t).$$
	Then it follows from  \eqref{eqfilter} and \eqref{eq19}  that
	\begin{equation} \label{eqxi}
		\begin{split}
			\dot \xi_1(t) =& -k_1|\xi_1(t)|^{1/2}\sgn\left(\xi_1(t)\right) -k_2 \xi_1(t) + \xi_2(t), \\ 
			\dot \xi_2(t)  =&  -v\sin \phi(t) \left( {\omega}(t) - \frac{v \sin \phi(t)} {d(t)}+ \frac{v_2(t)}{d(t)}\right) - \dot v_1(t) \\
			& - k_3\sgn\left( \xi_1(t) \right)- k_4 \xi_1(t),
		\end{split}
	\end{equation}
	where ${\omega}(t)= \omega_c+ {c_1}/{r_d}\cdot \left(\dot d(t) -\xi_2(t)\right) + c_2 \sat \left( e(t)\right).$
	Let $$f(\xi_1,\xi_2,t) = -v\sin \phi(t) \left({\omega}(t) - \frac{v \sin \phi(t)} {d(t)}+ \frac{v_2(t)}{d(t)}\right) - \dot v_1(t).$$
	If $d(t) > r_d$, it follows from \eqref{eq2} and \eqref{eq19} that
	\begin{align} \label{eqf}
		\left|f(\xi_1(t),\xi_2(t),t)\right| 
		& < \sigma_1 + \sigma_2 |\xi_2(t)| ,
	\end{align}
	where $\sigma_1 =\omega_c (2v+\bar v_o) + c_1 \omega_c(v+\bar v_o) +c_2 v + \bar a_o$ and $\sigma_2 =c_1 \omega_c $. 
	
	
	Let $\bm \xi(t) = \left[|\xi_1(t)|^{1/2}\sgn(\xi_1(t)), ~\xi_1(t), ~\xi_2(t)\right]'$ and consider the following Lyapunov function candidate 
	\begin{align*}
		V_{\Omega}(\bm \xi) = \bm \xi' \Omega \bm \xi ~\text{where}~\Omega = \frac{1}{2}\begin{bmatrix}
			4 k_4+k_1^2 & k_1 k_2 & -k_1 \\
			k_1 k_2 & 2 k_4 + k_2^2 & -k_2 \\
			-k_1 & -k_2 & 2
		\end{bmatrix}.
	\end{align*}
	Clearly, $V_{\Omega}(\bm \xi)$ is continuous, positive definite, and radially unbounded if $k_3 > 0$ and $k_4>0$, i.e., 
	\begin{align}\label{eqV}
		\lambda_{\min}(\Omega) \twon{\bm \xi}_2^2\le V_{\Omega}(\bm \xi) \le \lambda_{\max}(\Omega) \twon{\bm \xi}_2^2, 
	\end{align}
	where $\lambda_{\min}(\Omega)$ and $\lambda_{\max}(\Omega)$ are the minimum  and maximum eigenvalues of $\Omega$, respectively. 
	
	Taking the derivative of $V_{\Omega}(\bm \xi)$ along with \eqref{eqxi} leads to that
	\begin{align*}
		\dot V_{\Omega} (\bm \xi) &=- |\xi_1(t)|^{-1/2}\bm \xi' Q_1 \bm \xi -\bm \xi' Q_2 \bm \xi + 2 \xi_2(t) f(\xi_1,\xi_2,t) \\
		&~~~~- \left(k_2 \xi_1(t)  + k_1 |\xi_1(t)|^{1/2}\sgn(\xi_1(t))\right) f(\xi_1,\xi_2,t),
	\end{align*}
	where 
	\begin{align*}
		Q_1 &= \frac{k_1}{2} \begin{bmatrix}
			2 k_3+k_1^2 & 0 & -k_1 \\
			0 & 2k_4 + 5k_2^2 & -3k_1 \\
			-k_1 & -3k_2 & 1
		\end{bmatrix}, \text{and}\\
		Q_2 &= {k_2} \begin{bmatrix}
			k_3+2k_1^2 & 0 & 0 \\
			0 & k_4 + k_2^2 & -k_2 \\
			0 & -k_2 & 1
		\end{bmatrix}.
	\end{align*}
	Together with \eqref{eqf}, we have that 
	\begin{align*}
		&| 2 \xi_2(t) f(\xi_1,\xi_2,t) |\le 2 |\xi_2(t)|\left( \sigma_1 + \sigma_2 |\xi_2(t)|  \right) \\
		&~~~~\le  { \sigma_1 |\xi_1(t)|^{-1/2}} \left(  |\xi_1(t)|  +  \xi_2^2(t)  \right) + 2 \sigma_2 \xi_2^2(t),\\
		& |- k_2 \xi_1(t) f(\xi_1,\xi_2,t)|  \le k_2 | \xi_1(t)| \left( \sigma_1 + \sigma_2 |\xi_2(t)|  \right) \\
		&~~~~\le k_2 \sigma_1 | \xi_1(t)|  +  \left( k_2^2\xi_1^2(t) + \sigma_2^2\xi_2^2(t) \right)/2,\\
		& |-k_1 |\xi_1(t)|^{1/2}\sgn(\xi_1(t))f(\xi_1,\xi_2,t)| \\
		 &~~~~\le k_1 \sigma_1 |\xi_1(t) |^{-1/2}  |\xi_1(t) | ^2 + \left(k_1^2 |\xi_1(t)| + \sigma_2^2 \xi_2^2(t) \right)/2.
	\end{align*}
	Then, it immediately follows that
	\begin{align*}
		\dot V_{\Omega} (\bm \xi) \le  -{|\xi_1(t)|^{-1/2}} \bm \xi' (Q_1-Q_3) \bm \xi  -\bm \xi' (Q_2-Q_4) \bm \xi, 
	\end{align*} 
	where 
	\begin{align*}
		Q_3 &= \begin{bmatrix}
			(k_1+1) \sigma_1 & 0 & 0 \\
			0 & 0 & 0 \\
			0 & 0 & \sigma_1
		\end{bmatrix}, \\
		Q_4 &= \begin{bmatrix}
			k_2 \sigma_1 + k_1^2/2 & 0 & 0 \\
			0 & k_2^2/2      & 0 \\
			0 & 0 & 2 \sigma_2 +\sigma_2^2
		\end{bmatrix} .
	\end{align*}
	If $k_1 > 2 \sigma_1$, $k_2>0$, $k_3>\max\left\{0, ~(k_1 +1)\sigma_1/k_1 -k_1^2/2\right\} $, and $k_4>k_2^2(2k_1 + 5\sigma_1)/(k_1-2\sigma_2)$,  then $Q_1-Q_3$ is positive definite. Similarly, the conditions $k_1 >0$, $k_2 >\sigma_2^2 + 2\sigma_2 $, ~$k_3>\max\left\{0,~\sigma_1-2k_1^2 -k_1^2/(2k_2)\right\} $ and $k_4>\max\left\{0,~ k_2/2-k_2^2\right\}$ lead to that $Q_2-Q_4$ is positive definite. 
	
	Overall, the conditions on controller and filter parameters ensure that 
	\begin{align}\label{eqvdot}
		\dot V_{\Omega} (\bm \xi) &\le  -{|\xi_1(t)|^{-1/2}} \bm \xi' (Q_1-Q_3) \bm \xi  \nonumber\\
		 & \le  -{|\xi_1(t)|^{-1/2}}  \lambda_{\min} (Q_1-Q_3) \Vert\bm \xi \Vert_2^2.
	\end{align}
	
	In virtue of \eqref{eqV}, (\ref{eqvdot}), and the fact $|\xi_1(t)|^{1/2} \le \twon{\bm \xi }_2 \le {V_{\Omega}^{1/2}(\bm \xi)}/{\lambda_{\min}^{1/2}(\Omega)}$, it yields that 
	\begin{align*}
		\dot V_{\Omega} (\bm \xi) &\le  -\frac{\lambda_{\min}^{1/2}(\Omega)}{V^{1/2}(\bm \xi)}  \lambda_{\min} (Q_1-Q_3) \Vert\bm \xi \Vert_2^2\le -\gamma V_{\Omega}^{1/2} (\bm \xi),  
	\end{align*}
	where 
	\begin{align} \label{eqgamma}
		\gamma = {\lambda_{\min}^{1/2}(\Omega) \lambda_{\min} (Q_1-Q_3)}/ {\lambda_{\max}(\Omega)}.
	\end{align}
	By the comparison principle \cite[Lemma 3.4]{Khalil2002Nonlinear}, we obtain that 
	\begin{align*}
		\alpha_1(t) = d(t) ~\text{and}~\alpha_2(t) = \dot d(t), ~\forall t> t_0+T_2, 
	\end{align*}
	where $T_2= 2 V_{\Omega}^{1/2}(\bm \xi(t_0))/\gamma$ and $\gamma$ is given in \eqref{eqgamma}. Thus, the range-only controller in \eqref{eq2} is exactly identical to the range-based controller in \eqref{eq222} for all $t> t_0+T_2$ if $d(t_0+T_2) > 2r_d+(v+\bar v_o)T_1$. The rest of the proof follows from that of Proposition \ref{prop_moving}.  \end{proof}

\section{Simulations and Experiments} \label{secsim}

For brevity, we denote the states of the target and vehicle by $\bm s_o(t) = [\bm p'_o(t),\bm v_o(t)]'$ and $\bm s(t)= [\bm p'(t),\theta(t)]'$.   
To obtain the relative range in actual experiments, we respectively adopt an on onboard ultra-wideband (UWB) sensor in Section \ref{simu2} and use global positions to compute \eqref{eqrange} in Section \ref{simu1}. Moreover, we incorporate our proposed controller \eqref{eq2} into the control system of \cite{small} for a 6-DOF fixed-wing UAV and further take the noisy measurements into account in Section \ref{simu5} to validate the effectiveness of the SOSM filter \eqref{eqfilter}. Finally, comparisons with the existing methods are carried out in Section \ref{simu3}. 


\subsection{Experiments with a differential steering vehicle}\label{simu2}

\begin{figure}[t!]
	\centering{\includegraphics[width=0.7\linewidth]{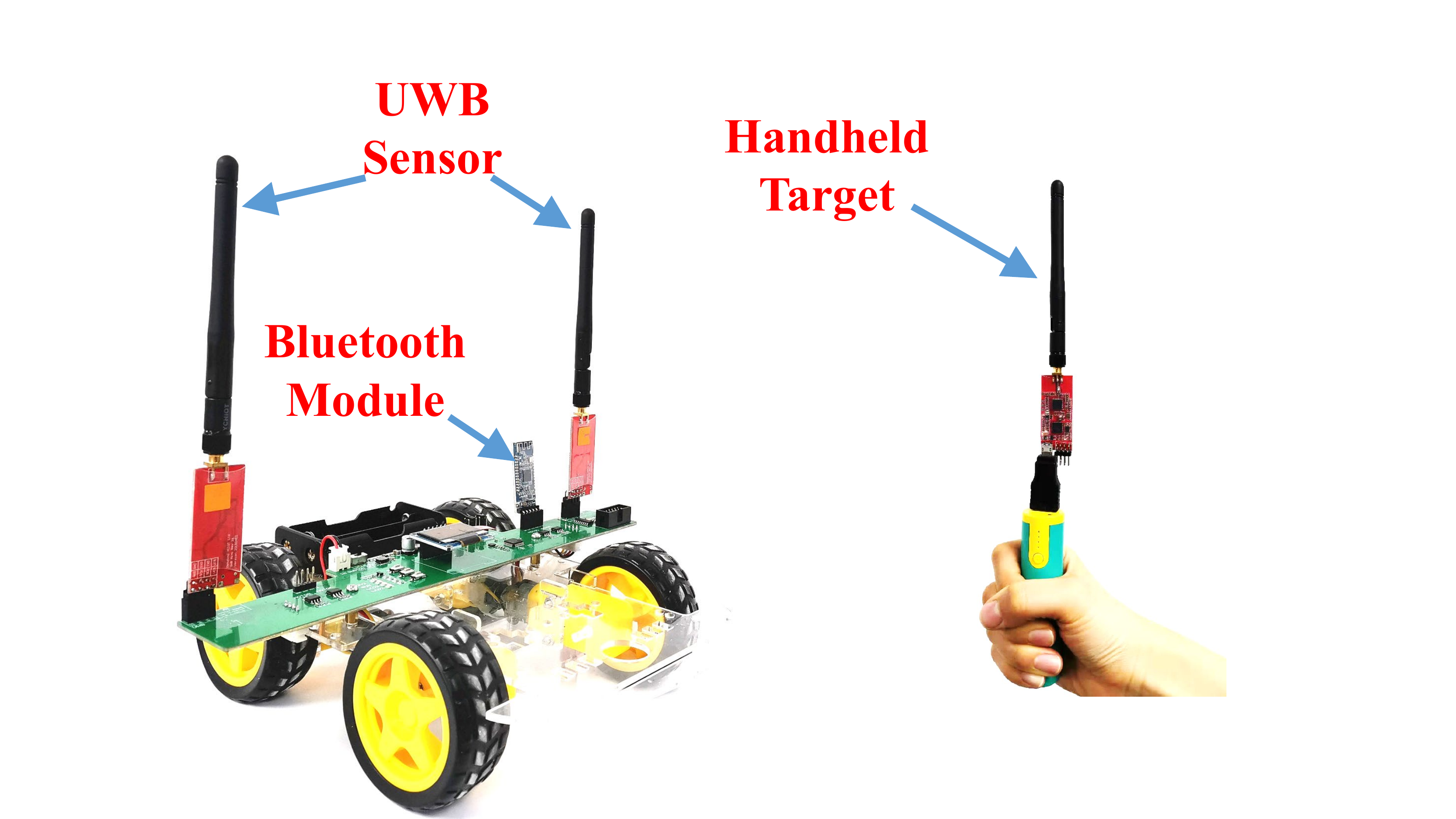}}
	\caption{The DSV equipped with an onboard UWB sensor and the handheld target to be tracked.}
	\label{fig25}
\end{figure}
\begin{figure}[t!]
	\centering{\includegraphics[width=0.8\linewidth]{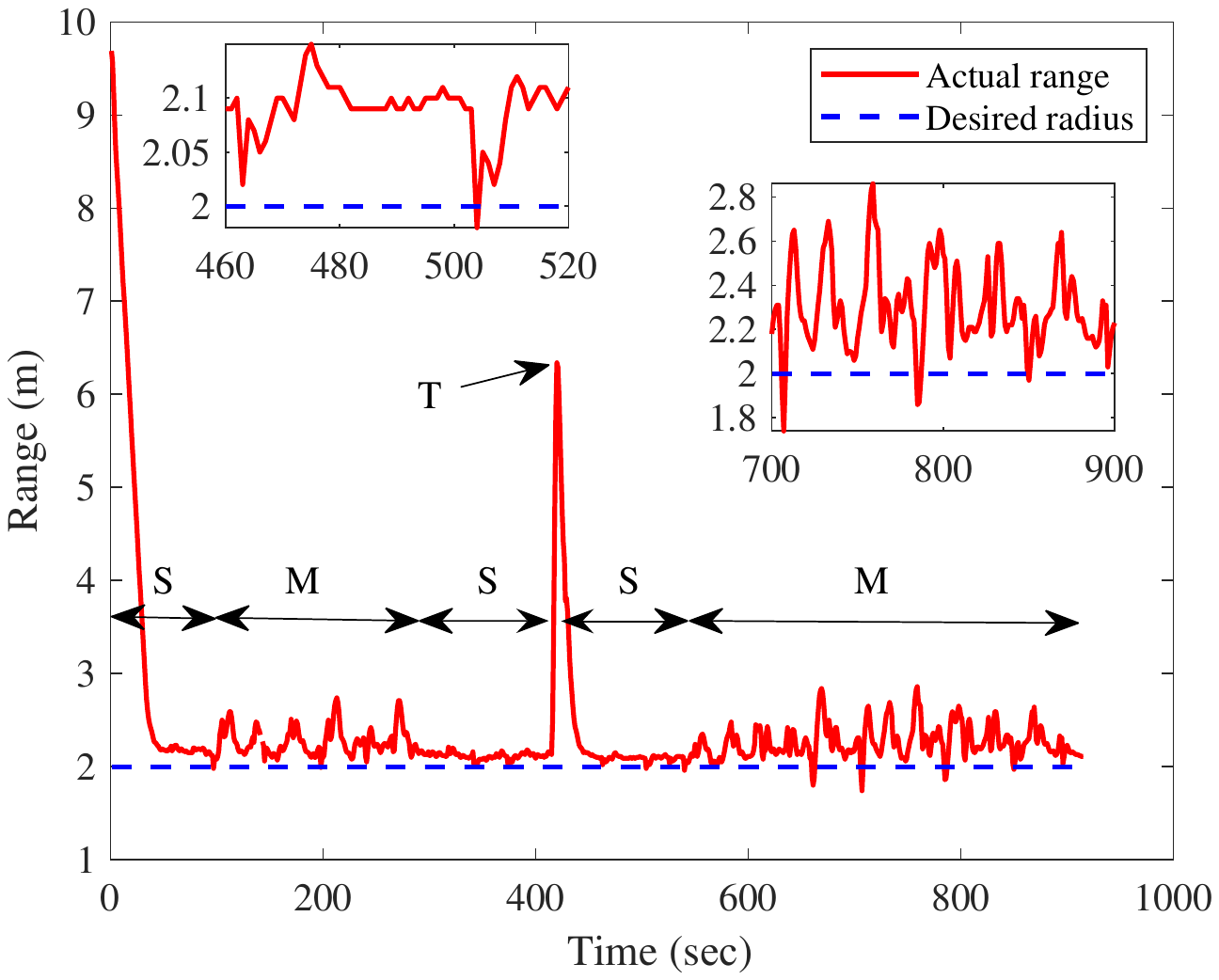}}
	\caption{Experimental result generated by the DSV, where ``S", ``M", and ``T" denote the different statuses of the handheld target with the following notations: stationary, slowly moved, and suddenly translocated, respectively.}
	\label{fig26}
\end{figure}

In this subsection, we adopt the differential steering vehicle (DSV) in Fig.~\ref{fig25} to test the proposed controller \eqref{eq222}, where an onboard UWB sensor measures the range to a handheld target at a frequency of 10 \si{Hz} and a Bluetooth module sends the real-time measurements at 1 \si{Hz}. The received data is visualized in Fig.~\ref{fig26}, from which we can observe that the DSV approaches the stationary target (denoted by ``S") from a far away position and then slides on a circular orbit between 0 and 100 \si{s}. Then the DSV can keep circumnavigating the target while it is slowly moved (denoted by ``M") from 100 to 270 \si{s} and from 560 \si{s} to the end. Although the target is suddenly translocated (denoted by ``T") at 420 \si{s}, the DSV immediately returns to the circular orbit with the new target position as its center and the original radius $r_d=2$ \si{m}. The above observations show the potential effectiveness of \eqref{eq222} in real applications.   

\subsection{Experiments with a Racecar and an omnidirectional vehicle}\label{simu1}

From Fig.~\ref{fig27}, a Racecar and a DJI Robomaster play the roles of tracker and target, respectively. A vision-based motion capture system is used to measure the current positions of the tracker and target by the markers at a frequency of 120 \si{Hz}. Then the positions are subsequently transformed into the relative range by \eqref{eqrange} as feedback information. The system can also record the real-time measurement to facilitate the experimental performance analysis. The Racecar has a servo motor with a maximum angle of $0.39$ \si{rad} to control its angular speed, the desired value of which is generated by our controller in \eqref{eq222}, while the target is omnidirectional and remotely manually operated through a mobile phone. 
Due to the space limitation, we set the constant linear speed of the Racecar as $v=1$ \si{m/s} and the predefined radius as $r_d=1$ \si{m}.

For the stationary Robomaster located at $\bm p_o=[-0.10, 1.85]'$ as Fig.~\ref{fig30}, the Racecar approaches the desired orbit from the initial position $\bm p(t_0)=[-0.01, -3.33]'$  and then slides on it. See Fig.~\ref{fig31} for an illustration. When the Robomaster is freely moving, both trajectories of the Racecar and Robomaster are exhibited in Fig.~\ref{fig28}(a), wherein the small square and circle respectively represent their initial positions $\bm p(t_0)=[-1.02, 1.60]'$ and $\bm p_o(t_0)=[1.26, -2.66]'$. Moreover, the relative trajectory of the Racecar with respect to the Robomaster is depicted in Fig.~\ref{fig28}(b). Furthermore, the tracking error is shown in Fig.~\ref{fig29}, and two supplementary videos are available at \cite{dong2020PD}. It is clear that the steady-state error has an upper bound and the objective \eqref{eqobj2} is achieved. 
%

\begin{figure}[t!]
	\centering{\includegraphics[width=0.8\linewidth]{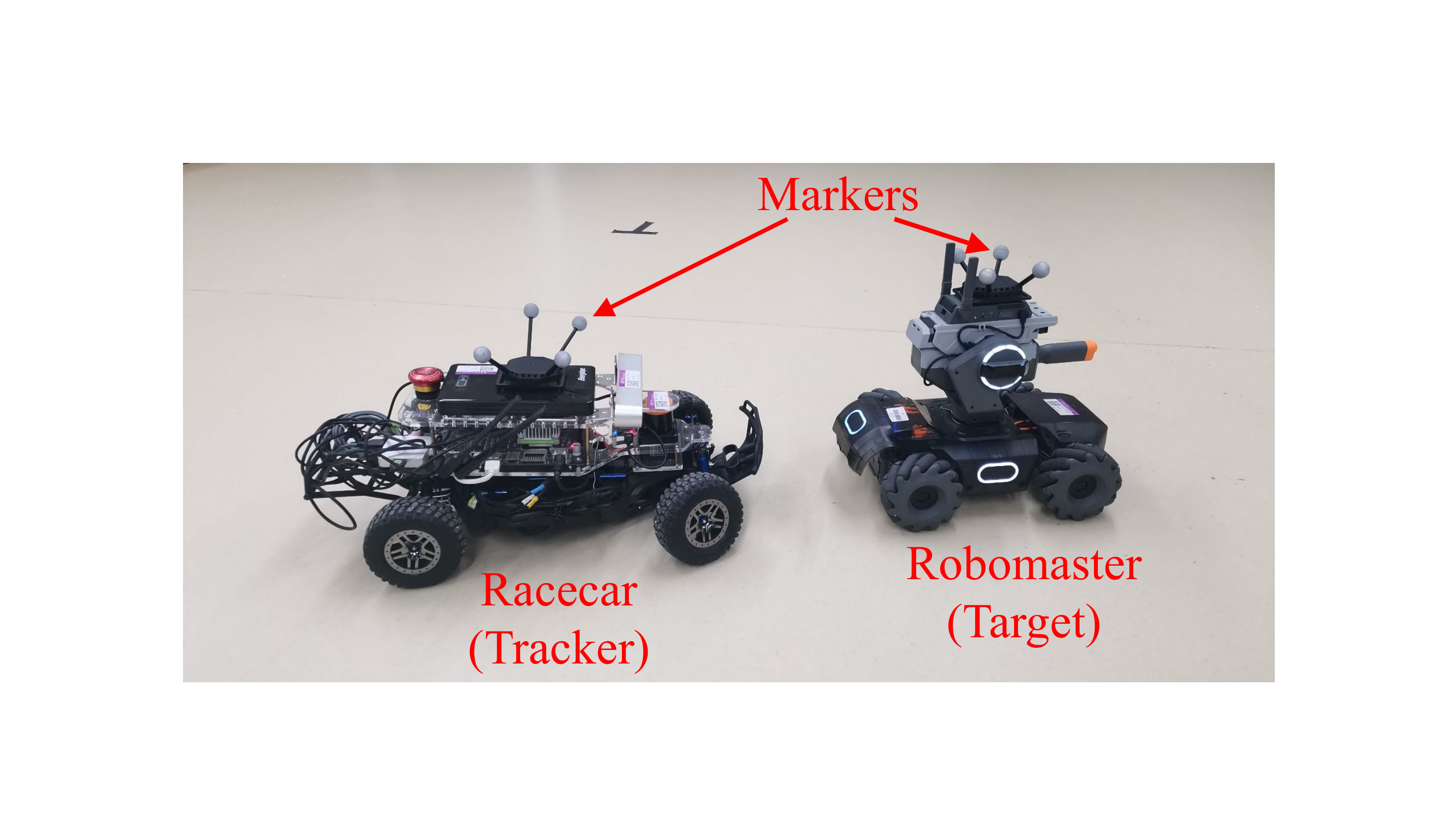}}
	\caption{The Racecar and Robomaster are adopted respectively to play the role of tracker and target.}
	\label{fig27}
\end{figure}

\begin{figure}[t!]
	\centering{\includegraphics[width=0.8\linewidth]{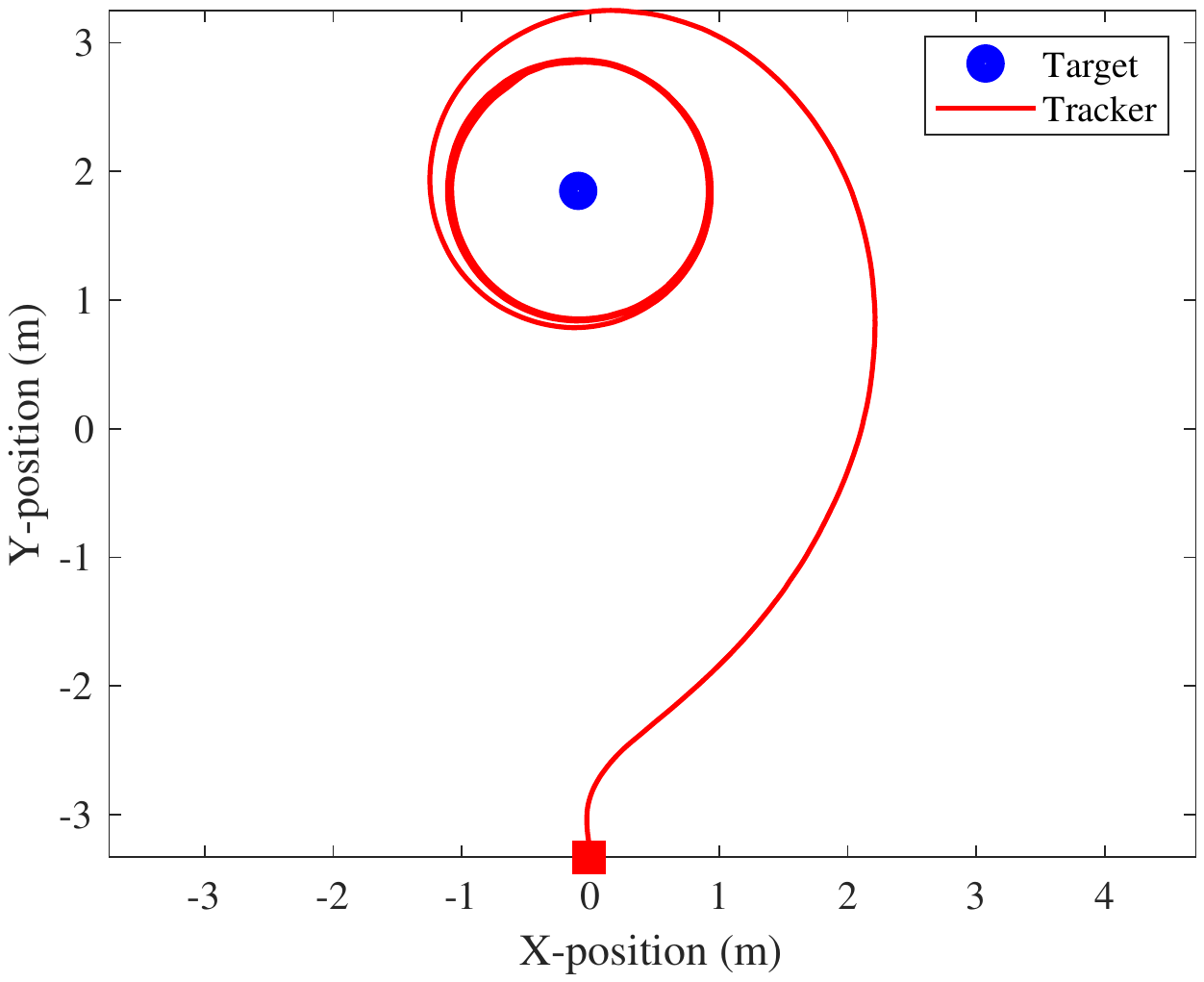}}
	\caption{Trajectory of the Racecar and position of the stationary Robomaster.}
	\label{fig30}
\end{figure}

\begin{figure}[t!]
	\centering{\includegraphics[width=0.8\linewidth]{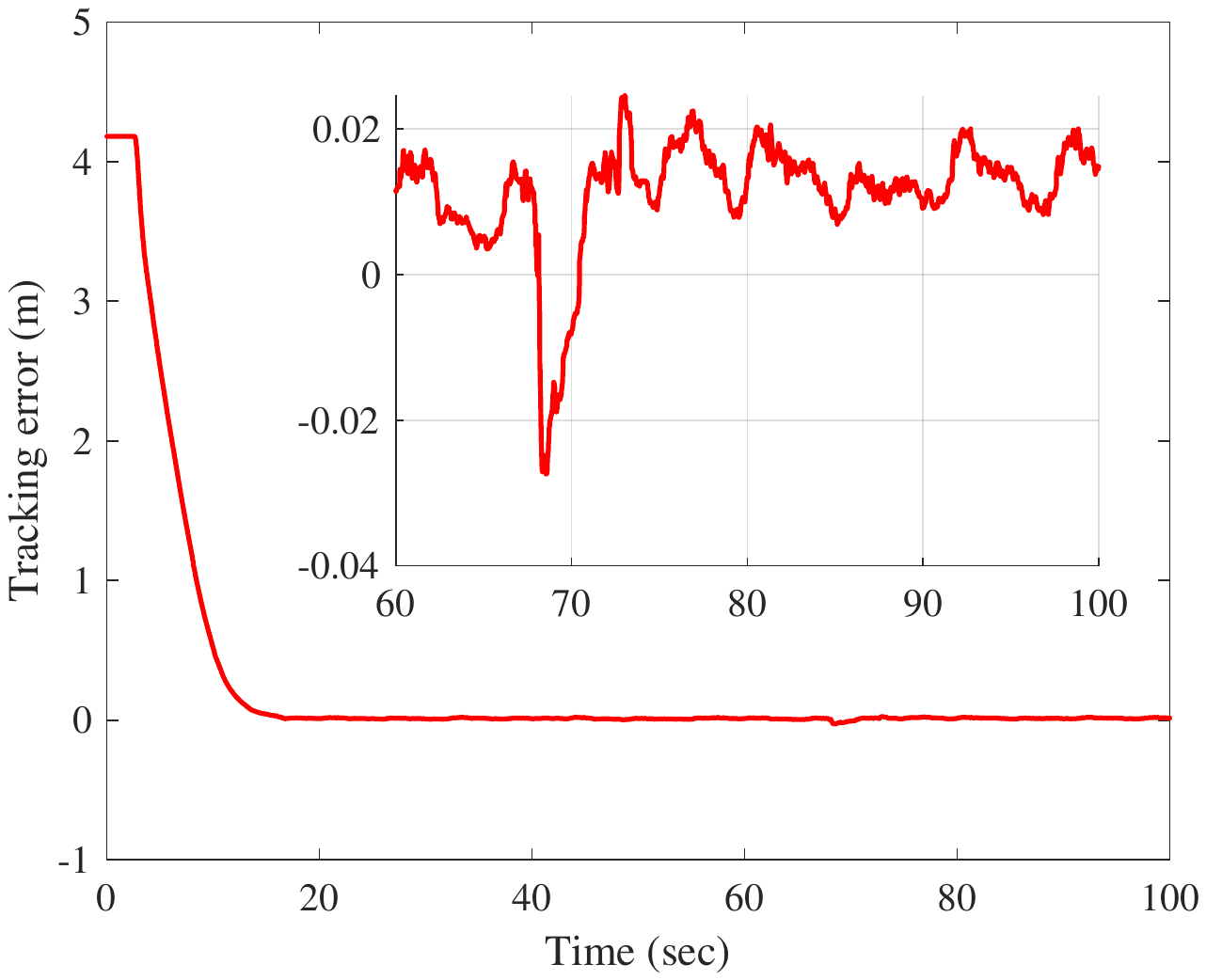}}
	\caption{Tracking error $d(t)-r_d$ of the Racecar.}
	\label{fig31}
\end{figure}

\begin{figure}[t!]
	\centering{\includegraphics[width=0.8\linewidth]{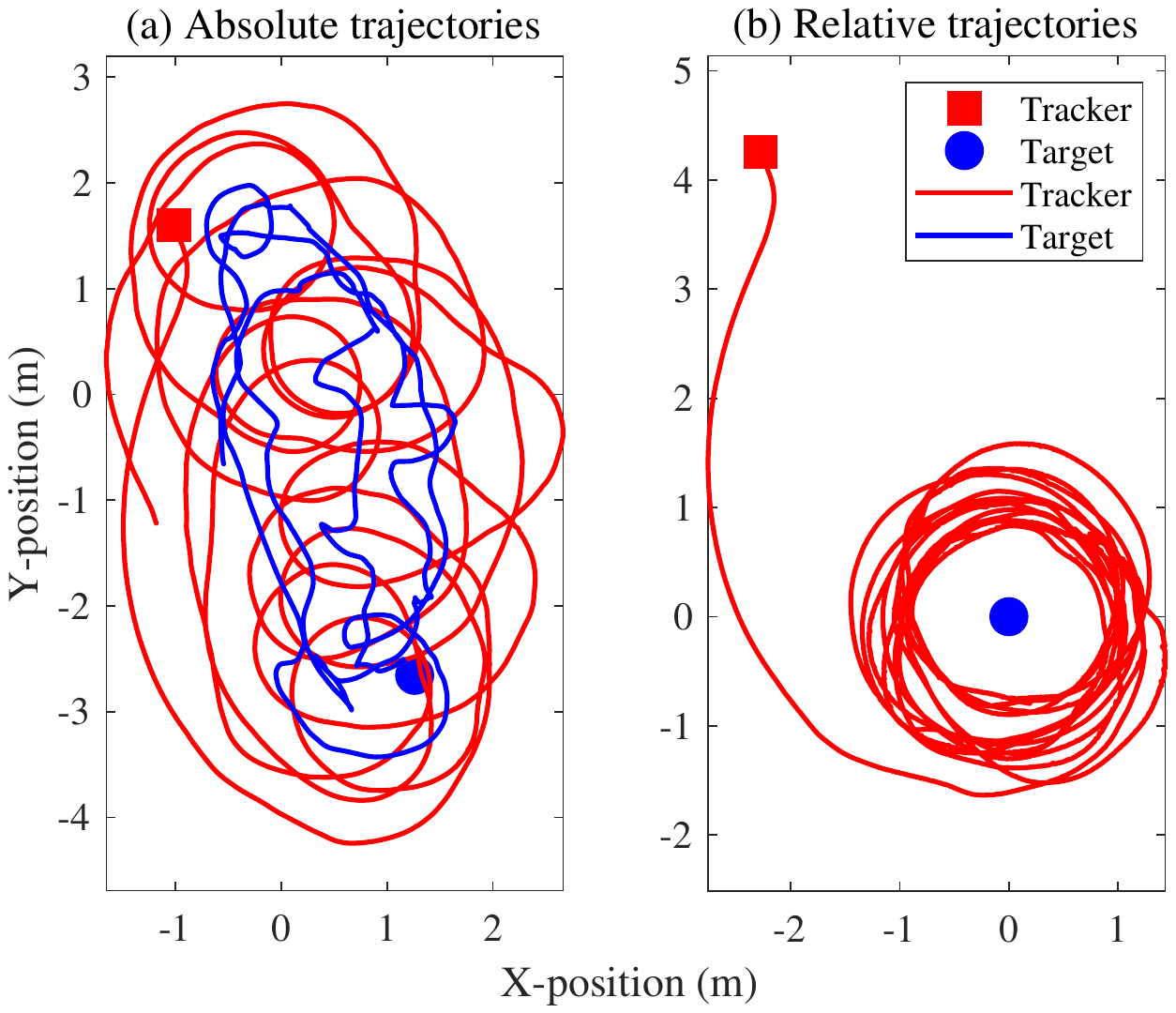}}
	\caption{Trajectories of the Racecar and Robomaster.}
	\label{fig28}
\end{figure}

\begin{figure}[t!]
	\centering{\includegraphics[width=0.8\linewidth]{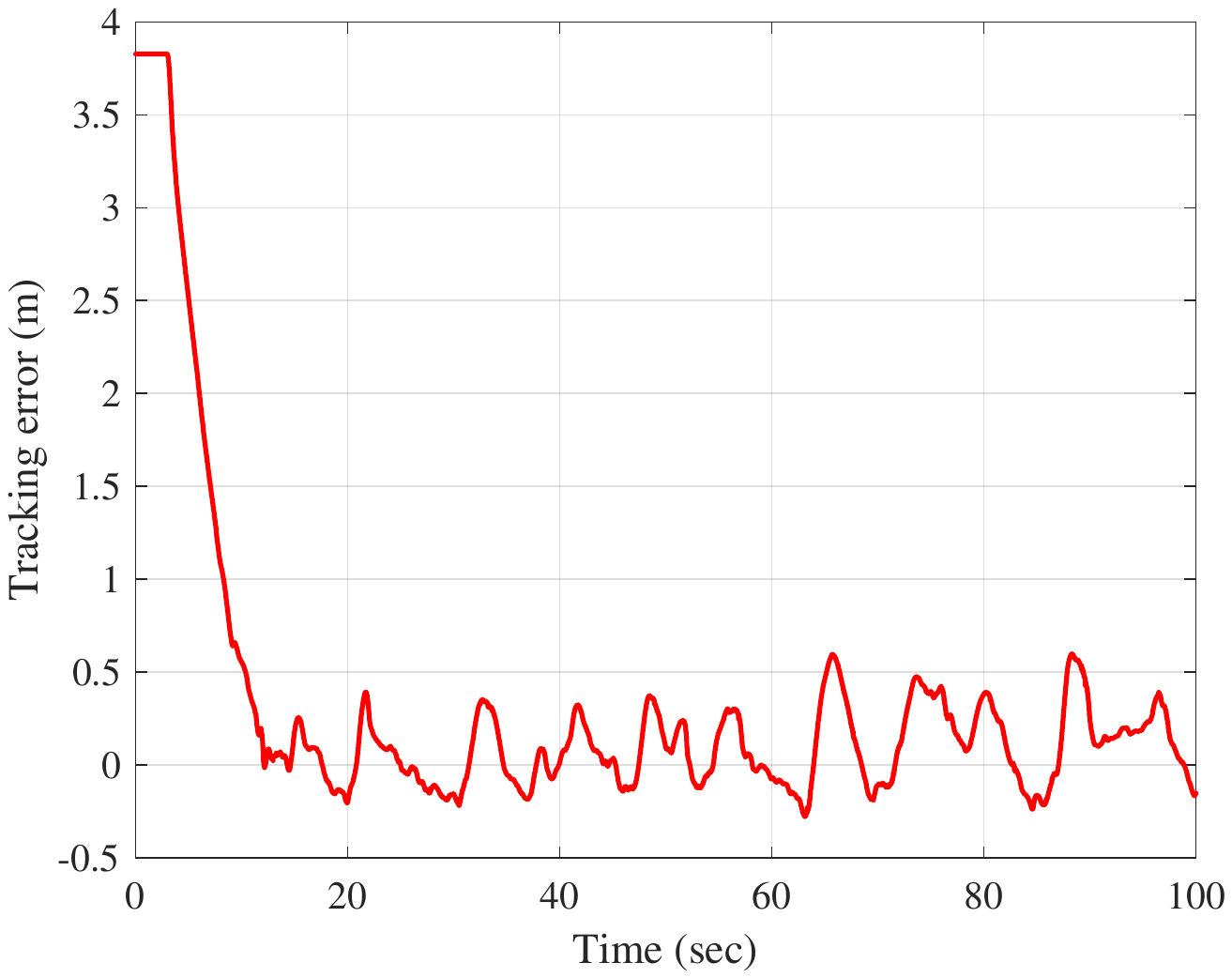}}
	\caption{Tracking error $d(t)-r_d$ of the Racecar.}
	\label{fig29}
\end{figure}

\subsection{Target circumnavigation by a fixed-wing UAV} \label{simu5}
\begin{figure}[t!]
	\centering{\includegraphics[width=1.0\linewidth]{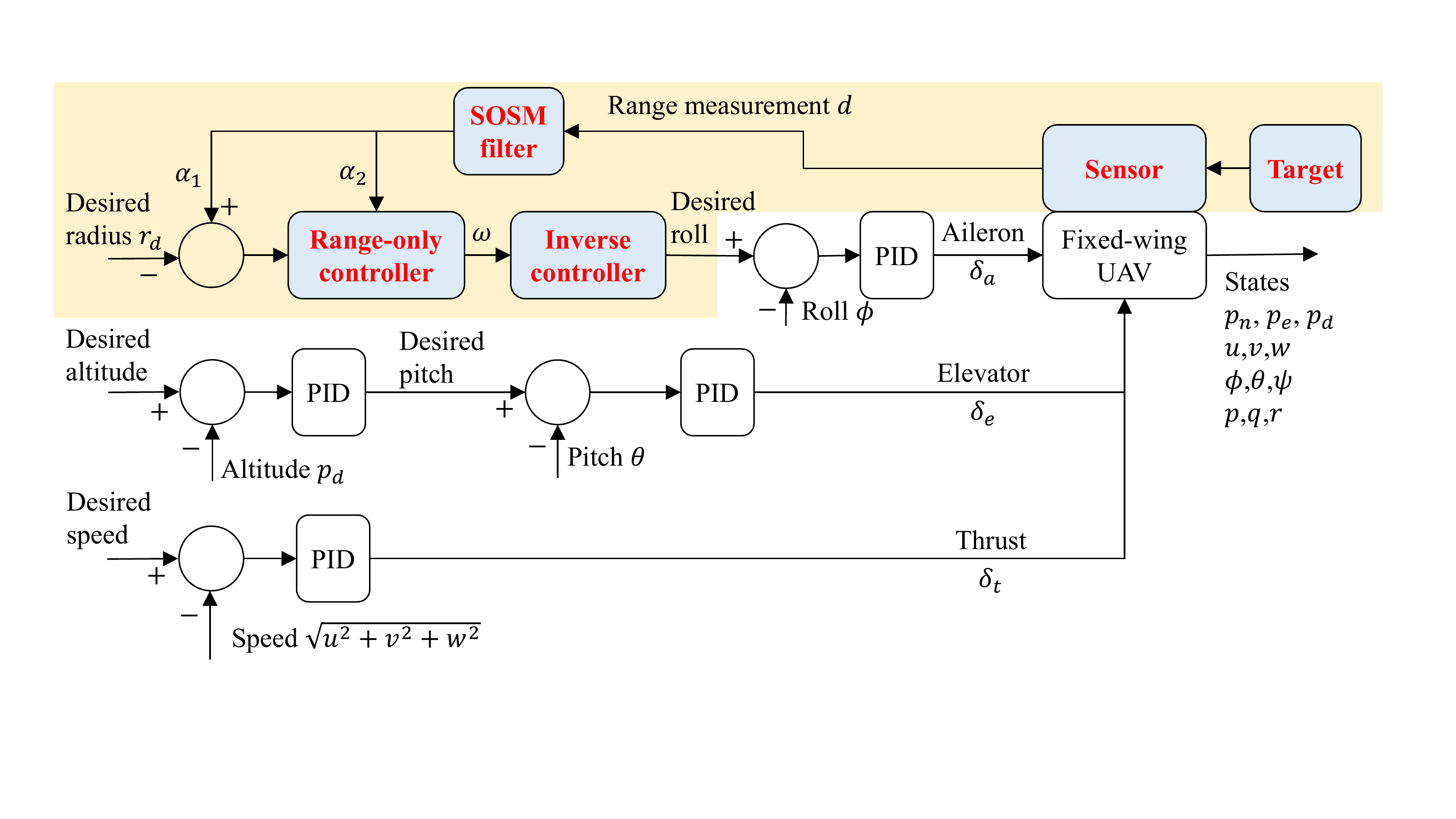}}
	\caption{Control architecture for the fixed-wing UAV.}
	\label{fig21a}
\end{figure}
In this subsection, a 6-DOF fixed-wing UAV \cite{Dong2019Flight,Beard2012Small} is adopted to test the effectiveness of the range-only controller in \eqref{eq2}. Due to the page limitation, we omit details of the complicated mathematical model of the UAV, which can be found in Chapter 3 of  \cite{Beard2012Small}, and directly adopt codes from \cite{small} for the model. The objective of circumnavigation requires the fixed-wing UAV to move at a constant altitude $-100$ \si{m} and forward speed $35$ \si{m/s} by adjusting the aileron deflection, elevator deflection, and propeller thrust. 

To this end, we adopt the control architecture of \cite{small} and modify it to verify our controllers. The details can be found in Fig.~\ref{fig21a} where the differences from \cite{small} are highlighted in the bright yellow shade. The inverse controller is designed to convert the desired angle speed $\omega(t)$ generated by \eqref{eq2} to the desired roll angle $\psi_{\text{des}}(t)$, which is given by
\begin{align*}
	\psi_{\text{des}}(t) = \tan^{-1}\left(  \frac{\omega(t)\sqrt{u^2+v^2+w^2}}{g}  \right) ,
\end{align*}
where $g$ is the acceleration of gravity. All the controllers in Fig.~\ref{fig21a} are encoded with saturation to simulate the physical characteristics of the UAV, whose values are the same as \cite{small}.


Moreover, we consider the situation that range measurements are corrupted by an additive Gaussian noise, i.e.,
\begin{align*}
	d(t) = \Vert \bm p(t) -\bm p_o(t) \Vert _2 + \eta(t),
\end{align*}
where $\eta(t) \sim \mathcal{N}(0,\sigma^2)$. 

The target in \eqref{eqtar} moves on the plane and its horizontal acceleration is generated by a uniform distribution, e.g. $\bm a_o(t) \sim \mathcal U[-1.0,1.0]$. The projected trajectories of both the target and UAV are given in Fig.~\ref{fig22}(a) with initial velocities $\bm v_o(t_0)=[4.0, 3.0]'$ \si{m/s} and $\bm v(t_0)=[35,0]'$ \si{m/s}, where the circle and square denote their initial positions $\bm p_o(t_0)=[0,100]'$ and $\bm p(t_0)=[0,0]'$, respectively. In addition, Fig.~\ref{fig22}(b) illustrates the relative trajectory of the UAV with respect to the target. Furthermore, both the actual range (rate) and its estimated version versus time are depicted by Fig.~\ref{fig24} with $r_d =400$ and $\sigma = 4$. One can observe from the dashed line in Fig.~\ref{fig24} that the maximum circumnavigation error is not larger than $6$ \si{m}, which implies that the performance of the proposed controller is not significantly degraded.


\begin{figure}[t!]
	\centering{\includegraphics[width=0.8\linewidth]{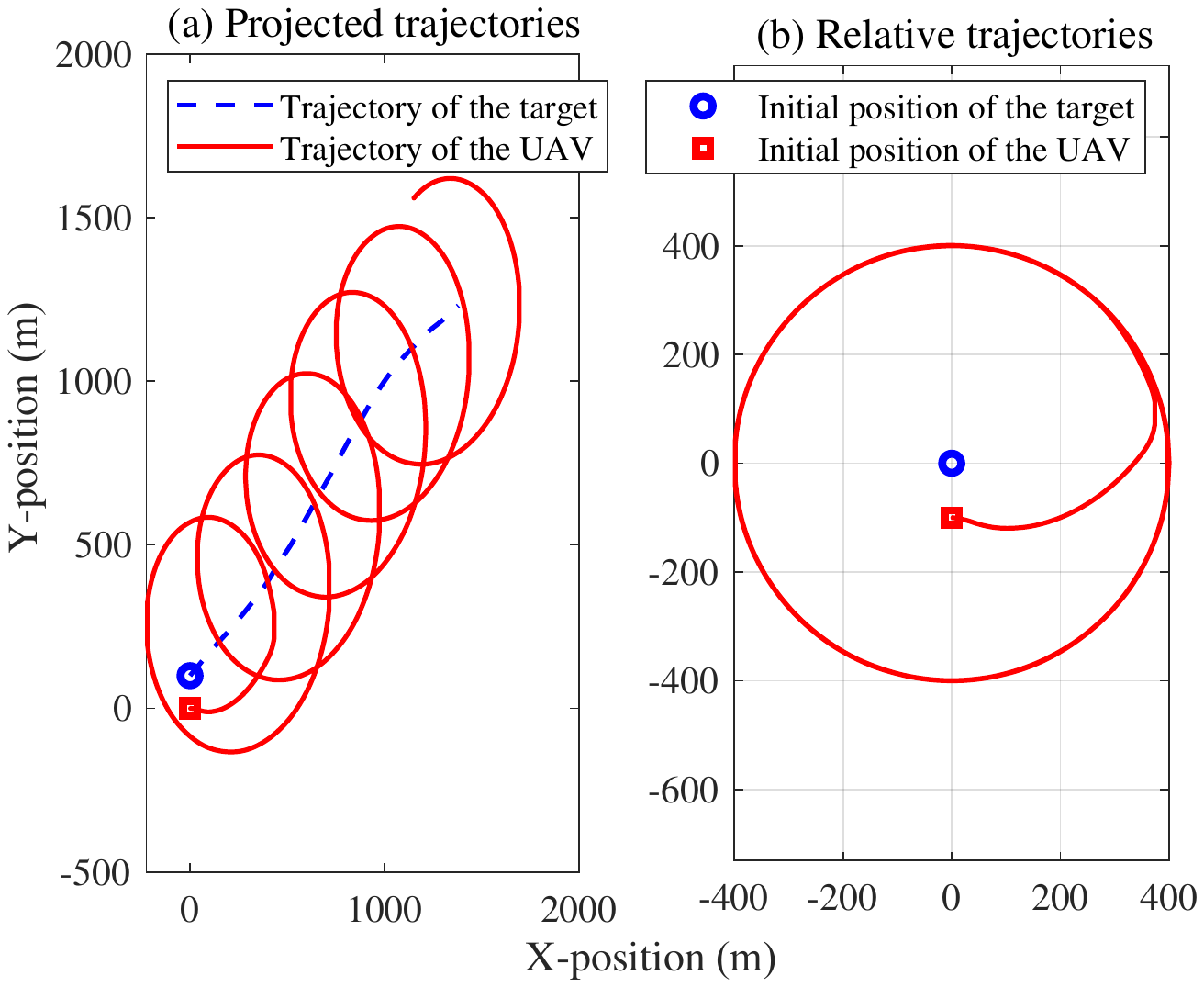}}
	\caption{Trajectories of the target and fixed-wing UAV.}
	\label{fig22}
\end{figure}
\begin{figure}[t!]
	\centering{\includegraphics[width=0.8\linewidth]{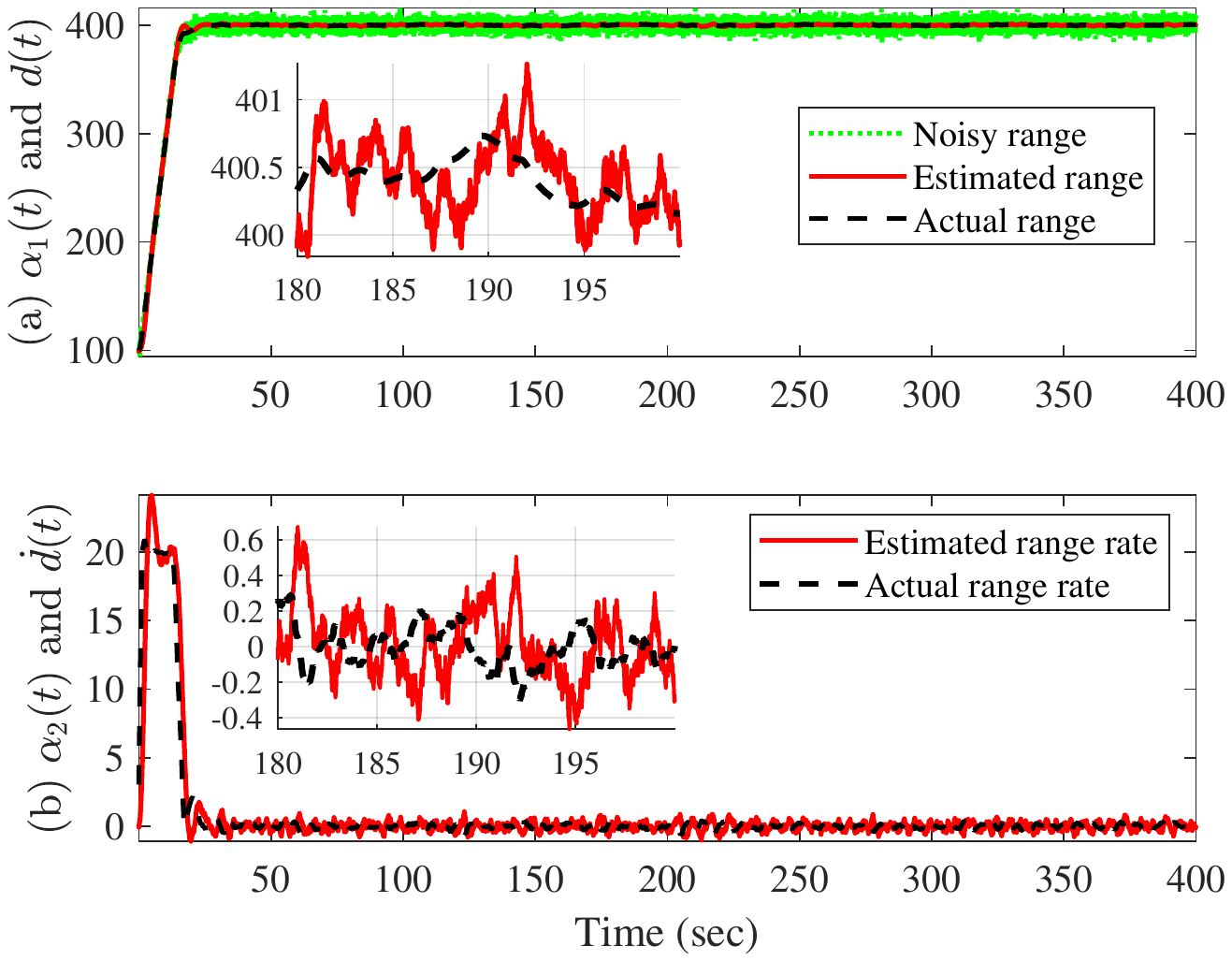}}
	\caption{Range $d(t)$ and its rate $\dot d(t)$ versus time with measurement noises.}
	\label{fig24}
\end{figure}

\subsection{Comparison with the existing methods} \label{simu3}

For comparison, we consider the constraint on control input and let $|\omega(t)| \le \bar \omega$, where $\bar \omega = 1$ \si{rad/s}  \cite{Matveev2011Range} in this subsection. The compared methods include the geometrical approach \cite{Cao2015UAV} with parameters $k=1$ and $r_a= \sqrt{3}$, the bearing approach \cite{zhang2020range} with parameter $k = 1.4/r_d$, the sliding mode approach \cite{Matveev2011Range} with $\delta=0.83$ and $\gamma = 0.3$, and the backstepping approach \cite{dong2019Target} with $k_1 =20$ and $k_2 =0.3$. 

When the target is stationary, the performance comparison is depicted in Fig.~\ref{fig14}, wherein $\bm s(t_0)= [3,3,0.25\pi]'$, i.e., $\phi(t_0)=0$ as shown by Fig.~\ref{fig14a}. It is observed from Fig.~\ref{fig14} that all methods other than the sliding mode approach can complete the task with zero steady-state error. In addition, the geometrical approach has large overshoots, and the convergence rate of the bearing approach is slowest. From Fig~\ref{fig15}, we observe that the input of the sliding mode approach switches between $-\bar \omega$ and $\bar \omega$, which results in the ``chattering" phenomenon as illustrated in Figs.~\ref{fig14} and \ref{fig14a}.

\begin{table}[!t]
	\caption{Parameters of the PD-like range-based controller \eqref{eq222}}
	\centering	
	\begin{tabular}{|c|c|c|}	
		\hline
		{Parameter}   &{$c_1$} &{$c_2$} \\
		\hline
		{Value}         & 200          & 30               \\       
		\hline
	\end{tabular}%
	\label{tab2}%
\end{table}%

Then, let the target be moving with $\bar v_o=0.15$ \si{m/s} and the acceleration $\bm a_o(t)$ in \eqref{eqtar} be generated by a sine distribution, e.g. 
$
\bm a_o(t) = [0.01\sin(0.01t), 0.01\cos(0.01t)]'.
$
 The control parameters are selected as those in Table \ref{tab2}. Fig.~\ref{fig12} illustrates the results with $\bm s_o(t_0)=[0,0,-\sqrt{2}\bar v_o/2,-\sqrt{2}\bar v_o/2]'$ and $\bm s(t_0)= [5,0,-0.6\pi]'$. Since both the geometrical approach and the bearing approach are designed for the stationary target, they cannot handle the case of a moving target. The performance of our controller is similar to that of the sliding mode approach. However, the mean-square steady-state circumnavigation error (MSSE)\footnote{$\text{MSSE} = \frac{1}{n} \sum_{i=1}^{n} (d(i)-r_d)^2$ where $i$ denotes the $i$-th time step.} of our PD-like controller is less than that of the sliding mode approach in the time interval from $160$ \si{s} to $200$ \si{s}. See the partially enlarged view of Fig.~\ref{fig12}. Moreover, Fig.~\ref{figinput} confirms that the controller limit only takes effect at the early stage of the system.


Overall, the range-based controller in (\ref{eq222}) outperforms the methods in  \cite{Cao2015UAV,zhang2020range,Matveev2011Range}. Particularly, our method is effective in handling the problem of moving target circumnavigation.

\begin{figure}[t!]
	\centering{\includegraphics[width=0.8\linewidth]{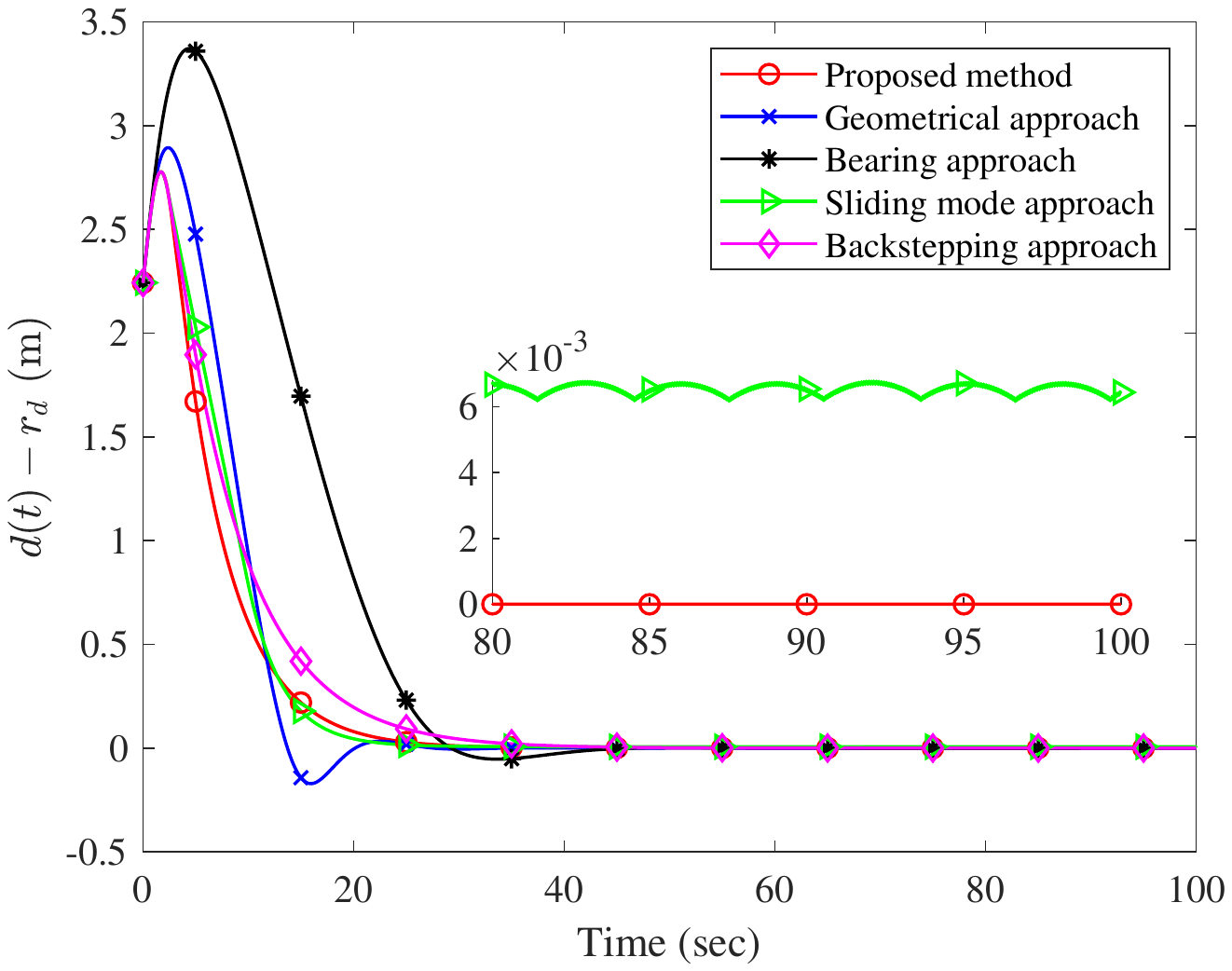}}
	\caption{Tracking errors $d(t)-r_d$ when the target is stationary.}
	\label{fig14}
\end{figure}
\begin{figure}[t!]
	\centering{\includegraphics[width=0.8\linewidth]{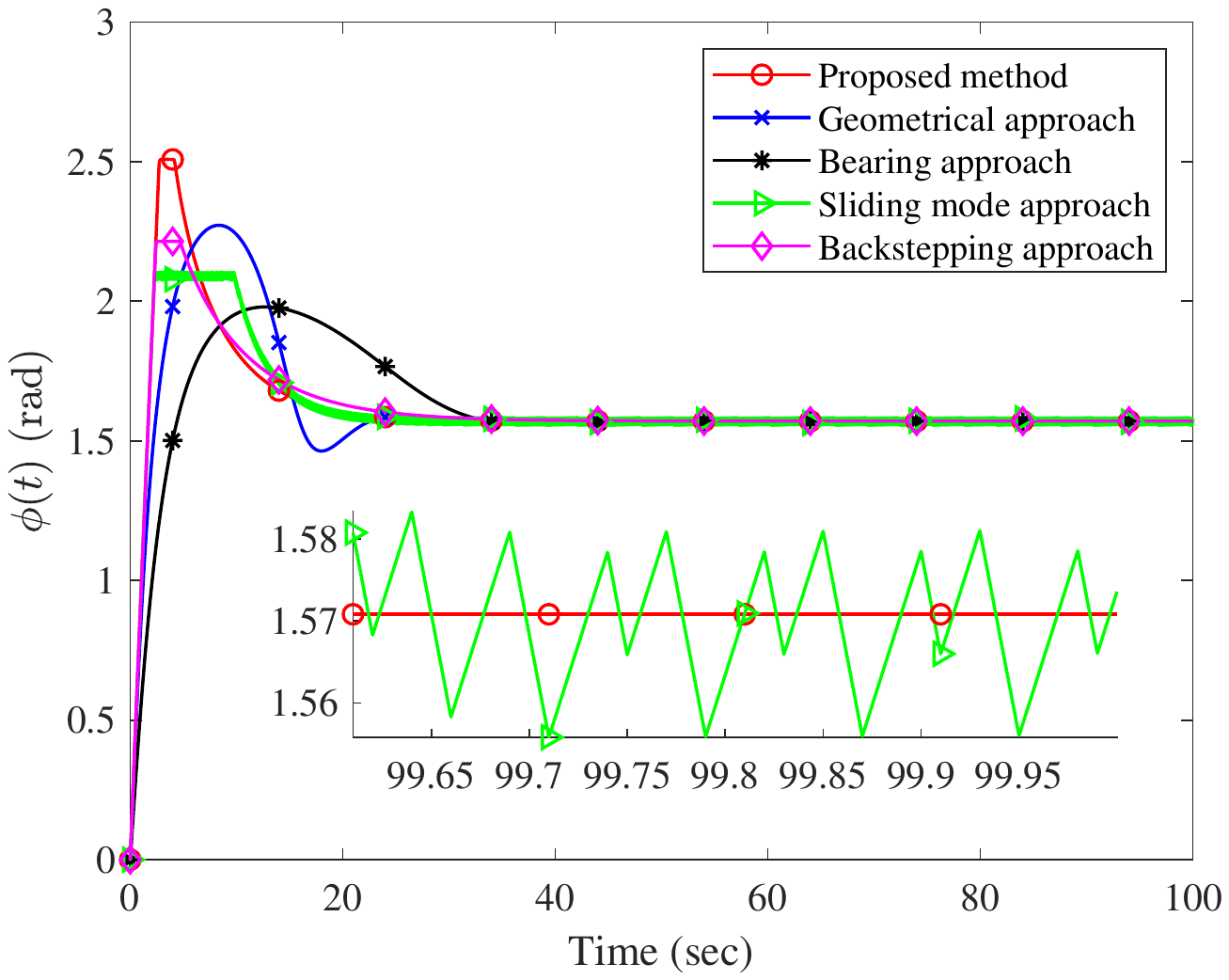}}
	\caption{Variations of $\phi(t)$ when the target is stationary.}
	\label{fig14a}
\end{figure}

\begin{figure}[t!]
	\centering{\includegraphics[width=0.8\linewidth]{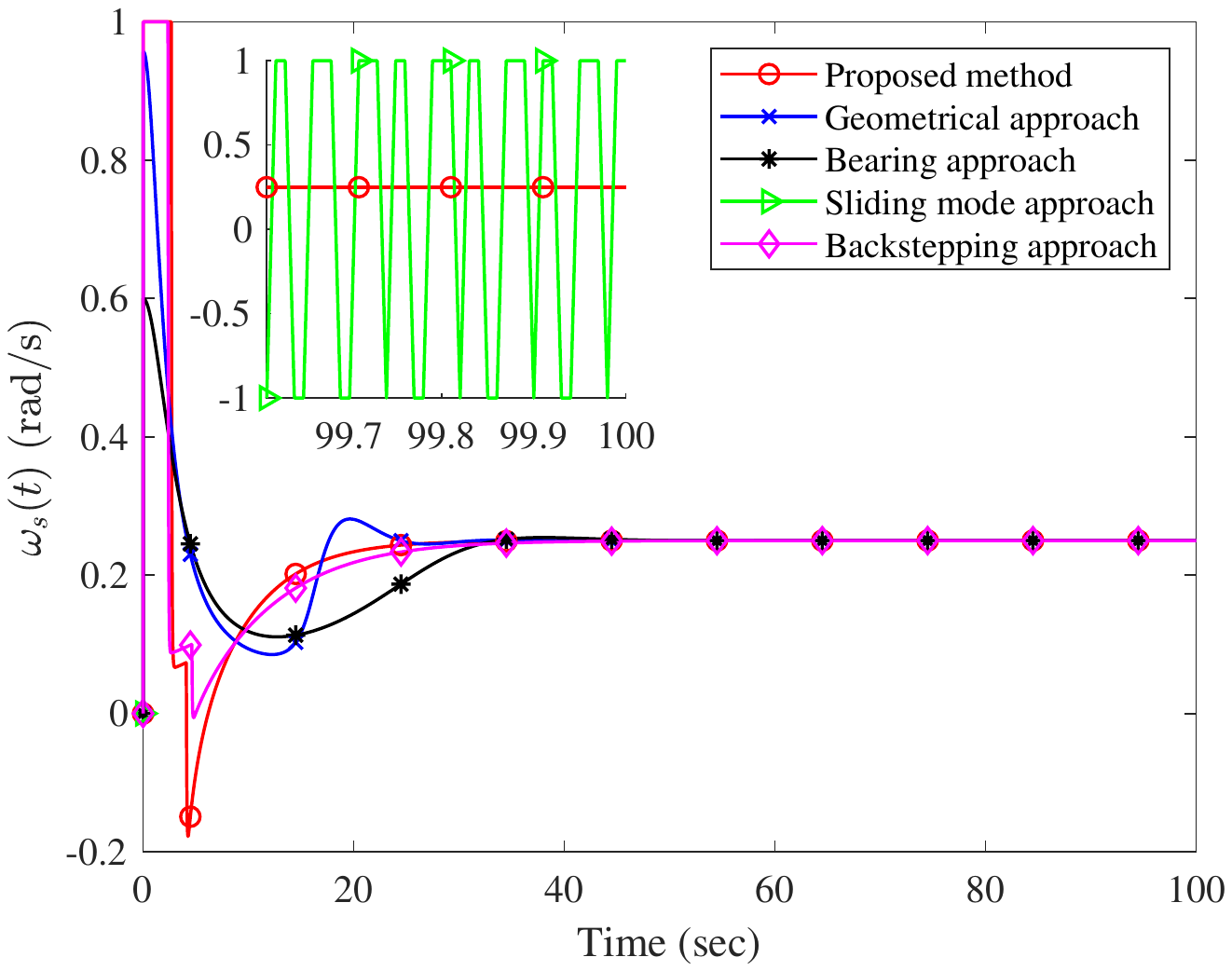}}
	\caption{Control inputs $\omega_s(t)$ when the target is stationary.}
	\label{fig15}
\end{figure}

\begin{figure}[t!]
	\centering{\includegraphics[width=0.8\linewidth]{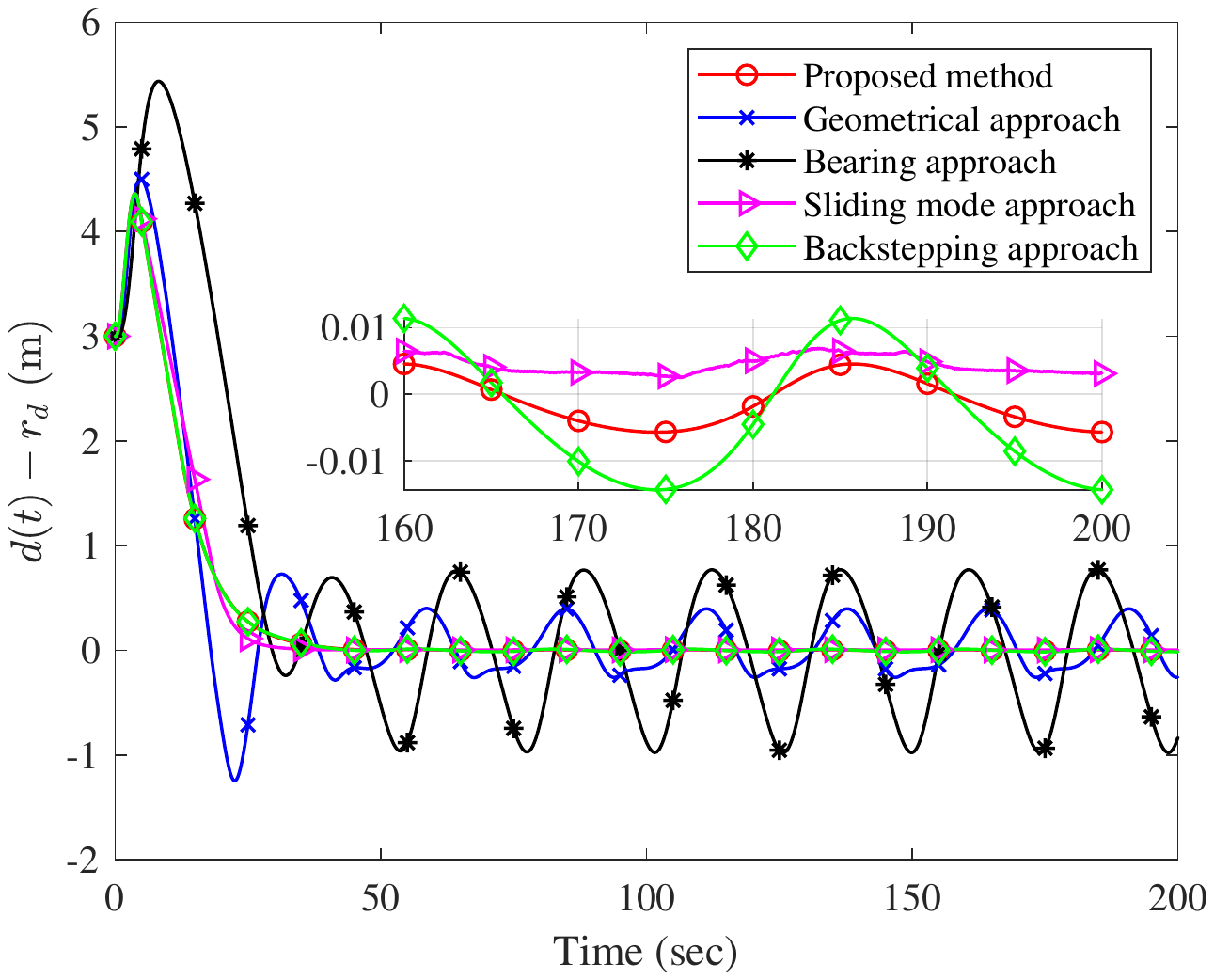}}
	\caption{Tracking errors $d(t)-r_d$ when the target is moving.}
	\label{fig12}
\end{figure}

\begin{figure}[t!]
	\centering{\includegraphics[width=0.8\linewidth]{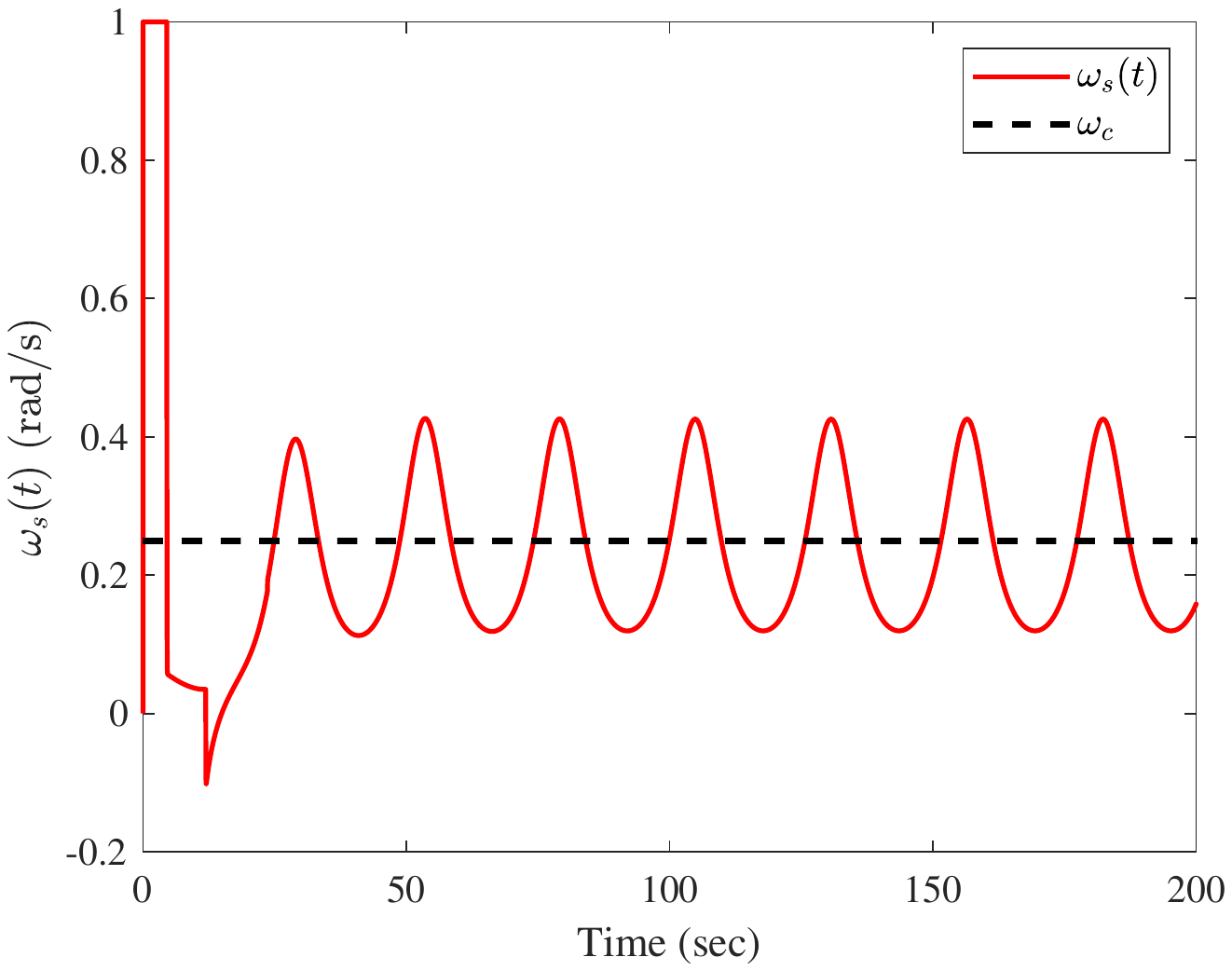}}
	\caption{Control input $\omega_s(t)$ with the limit $\bar \omega=1$ \si{rad/s}.}
	\label{figinput}
\end{figure}

\section{Conclusion} \label{sec6}

In this paper, we have proposed a range-only controller to drive a nonholonomic vehicle to circumnavigate a moving vehicle with double-integrator kinematics which plays the role of a target. Given that both the range and range rate measurements are known, the proposed controller has a Proportional Derivative (PD)-like form with a bias to eliminate steady-state circumnavigation error. Thus, for a stationary target, the controller can ensure global convergence and local exponential stability near the equilibrium with zero steady-state error. Moreover, we explicitly showed that the upper bound of the circumnavigation error is proportional to the maximum linear speed and acceleration of the target. Furthermore, we revised the range-based controller by replacing the actual range rate with its estimated version by designing a second-order-sliding-mode (SOSM) filter.  Finally, the numerical simulations and real experiments validated our theoretical results, and showed that our method outperforms the existing controllers and is particularly effective in the case of maneuvering target.



%

%

%

%
%
%

\appendix
\subsection{Proof of Proposition 1} \label{suba1}
To prove Proposition \ref{prop2}, we first show that there exists a finite time instant $t_1\ge t_0$ such that $\phi(t) \in [0,\pi], \forall t\ge t_1$, for any initial state, see Lemma \ref{lemma1}. Then,  the closed-loop system in (\ref{eq5}) under (\ref{eq222}) is shown to be asymptotically stable with an exponential convergence rate. 

\begin{lemma} \label{lemma1}
	Under the conditions in Proposition \ref{prop2}, 
	there exists a finite time instant $t_1\ge t_0$ such that $\phi(t) \in [0,\pi]$, $\forall t \ge t_1$, for any initial state $\phi(t_0) \in (-\pi,\pi]$.
\end{lemma}
\begin{proof}
	Inserting (\ref{eq222}) into (\ref{eq5}) yields that
	\begin{equation}\label{eq36}
	\begin{split}
	\dot d(t) &= v \cos \phi(t),\\
	\dot \phi(t) &= \omega_c + \frac{c_1 }{r_d}  \dot d(t) + c_2 \sat \left( \frac{d(t)-r_d}{r_d}\right)  - \frac{v \sin \phi(t)}{d(t)}. 
	\end{split}
	\end{equation}	
	
	If $d(t_0)=0$, the vehicle will immediately leave the target due to the constant linear speed $v$, i.e., $d(t_0^+)>0$. By Definition \ref{defi},  $\phi(t_0^+)=0$, which jointly with \eqref{eq36} implies that 
	\begin{align*}
		\dot \phi(t_0^+) > (c_1+1)\omega_c - c_2>0. 
	\end{align*}
	This results in that $\phi(t)$ enters the interval $(0,\pi/2)$ and $d(t)$ increases after $t_0$. Henceforth, we only consider the case of $d(t)>0$.
	
	Moreover, if $\phi(t) = 0$, it follows from (\ref{eq36}) that  
	\begin{align}  \label{eq35}
		\dot \phi(t)
		> (1 +  c_1) \omega_c-c_2 > 0.
	\end{align}
	Similarly, $\phi(t) = \pi$ leads to that
	\begin{align} \label{eq33}
		\dot \phi(t) \le  (1 -  c_1) \omega_c+ c_2 < 0.
	\end{align}
	
	Together with the fact that $\dot \phi(t)$ is continuous with respect to $t$ when $d(t)>0$, it implies that $\phi(t)\in [0,\pi]$ for all $t\ge t_0$ if $\phi(t_0)\in [0,\pi]$. 
	
	Next, we only need to show that there exists a finite time instant $t_1 > t_0$ such that $\phi(t_1) \in [0,\pi]$ if $\phi(t_0)\in (-\pi,0)$. To this end, three cases are considered.
	\begin{itemize}
		\item [(a)] $d(t_0) \in [r_d,\infty)$ and $\phi(t_0) \in [-\pi/2,0)$.
		\item [(b)] $d(t_0) \in [r_d,\infty)$ and $\phi(t_0) \in (-\pi,-\pi/2)$.
		\item [(c)] $d(t_0) \in (0, r_d)$ and $\phi(t_0) \in (-\pi,0)$.
	\end{itemize}
		
	
	{\bf Case (a)}: if $d(t) \in [r_d,\infty)$ and $\phi(t) \in [-\pi/2,0)$, it follows from (\ref{eq36}) that $\dot d(t) \ge 0$ and $\dot \phi(t) >\omega_c$. That is, $d(t)$ is increasing and $\phi(t)$ is strictly increasing with a rate greater than $\omega_c$. Then, Case (a) cannot always hold and must be violated in the sense that there exists a finite time instant $t_1>t_0$ such that  $\phi(t_1)\ge 0$.

	
	\begin{figure}[t!]
		\centering{\includegraphics[width=0.8\linewidth]{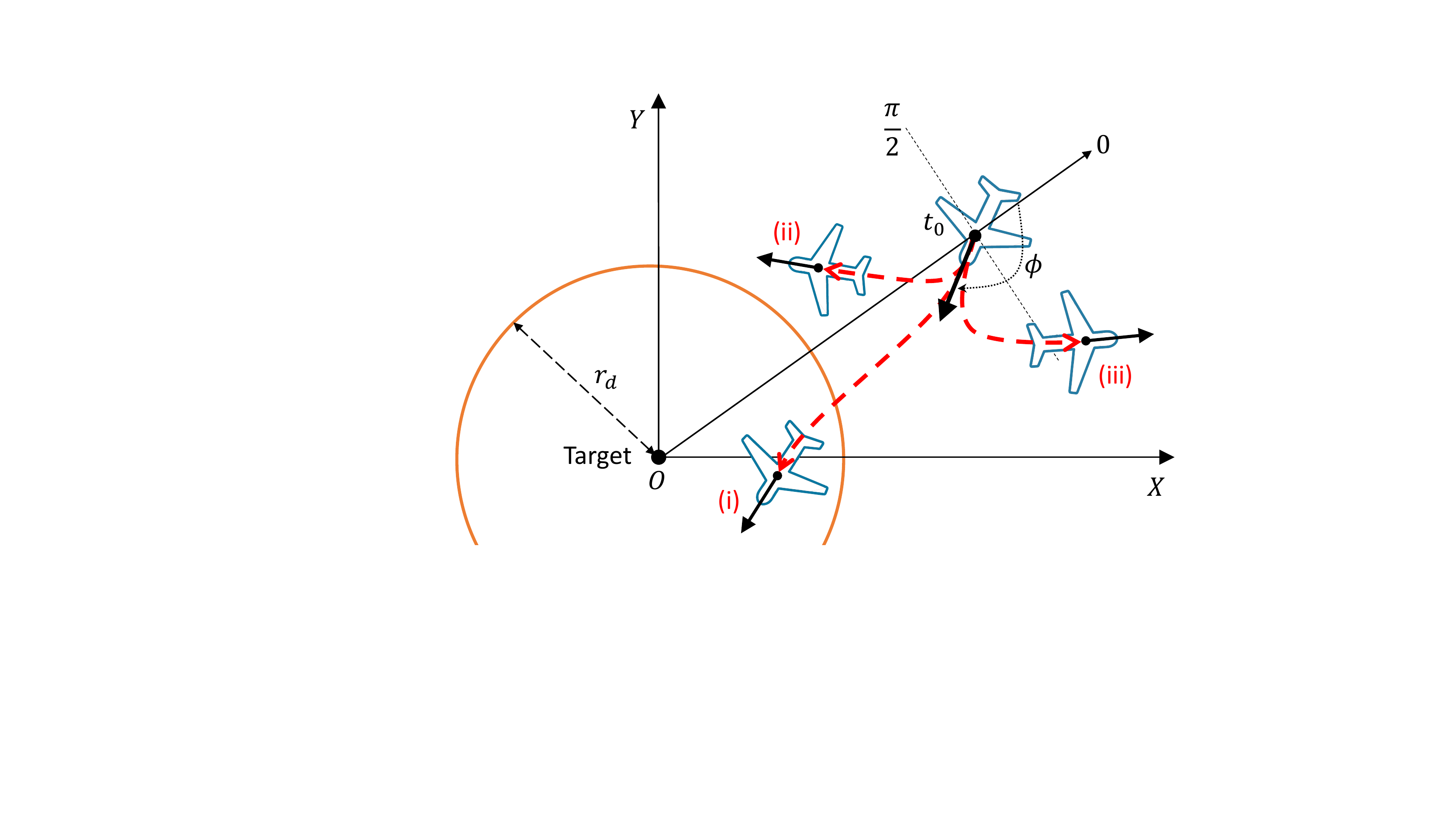}}
		\caption{Illustration  of Case (b), i.e., $d(t_0) \in [r_d,\infty)$ and $\phi(t_0) \in (-\pi,-\pi/2)$.}
		\label{fig16}
	\end{figure}	
	{\bf Case (b)}: by (\ref{eq36}), (\ref{eq33}) and $d(t)\ge r_d$, it holds that 
	\begin{eqnarray}
			&& \dot \phi(t)  > {\omega_c}, ~\text{if}~ \phi(t) = -\pi/2, \label{casenegative} \\
			&&\dot \phi(t)  < {0},~~~\text{if}~\phi(t) = -\pi.\label{caseneg}
	\end{eqnarray}
	We use Fig.~\ref{fig16} to help illustrate this case. By \eqref{casenegative} and the continuity of $\dot \phi(t)$, there is a sufficiently small $\varepsilon_0>0$ such that $\dot \phi(t) > 0$ if $|\phi(t) +\pi/2|\le\varepsilon$ for any $\varepsilon\in(0,\varepsilon_0)$. 
	
	If $\phi(t)\in(-\pi,-\pi/2-\varepsilon]$, it follows from \eqref{eq36} that $\dot d(t)\le -v\sin\varepsilon<0$. Thus, either there exists some finite $\delta>0$ such that $d(t_0+\delta)< r_d$, which corresponds to {\bf Case (c)}, or $\phi(t)$ will escape from the interval $(-\pi,-\pi/2-\varepsilon]$. In the later case, there exists a finite $\delta>0$ such that $\phi(t_0+\delta) \ge -\pi/2-\varepsilon $ or $\phi(t_0+\delta)\le -\pi$. If $\phi(t_0+\delta) \ge -\pi/2-\varepsilon$,  letting $\varepsilon$ go to zero reduces to either {\bf Case (a)} or {\bf Case (c)}, depending on the value of $d(t_0+\delta)$. If $\phi(t_0+\delta) \le -\pi$, it means that $\phi(t)$ must have already entered the interval $[0, \pi]$.

	{\bf Case (c)}: we first show that there is no stable equilibrium of \eqref{eq36}. Letting $\phi(t) = -\pi/2$, it follows from (\ref{eq36}) and $d(t)\in(0,r_d)$ that 
	\begin{align}\label{eqa1}
		\dot \phi(t) 
		= \frac{c_2 }{  r_d d(t) } \left(d^2(t) + \frac{v -c_2  r_d}{c_2 } d(t) + \frac{v r_d }{c_2 }\right).
	\end{align}
	Let $\triangle =({v}/{c_2}-r_d)^2-4vr_d/c_2$. 
		\begin{itemize}
		\item If $\triangle<0$, it follows from \eqref{eqa1} that $\dot \phi(t) > 0$, i.e., there is no equilibrium in Case (c). 
		\item If $\triangle\ge 0$, letting $\dot \phi(t)=0$ in \eqref{eqa1} yields two solutions with respect to $d(t)$, which are respectively denoted by $d_1$ and $d_2$. Then, $\bm x_{\text{eq}1}=[d_1,-\pi/2]'$ and $\bm x_{\text{eq}2}=[d_2,-\pi/2]'$ are two equilibria of \eqref{eq36}. By linearizing \eqref{eq36} at $\bm x_{\text{eq}i}$, $i\in\{1,2\}$, it follows that
		\begin{align*}
		\dot {\bm x}(t)=\begin{bmatrix}
		0 & v   \\{ c_2}/{r_d} -{ v }/{d_i^2} &  {c_1 \omega_c}
		\end{bmatrix} (\bm x(t) -\bm x_{\text{eq}i}).
		\end{align*} 
		Since $c_1 \omega_c>0$, one can easily verify that the transition matrix of the above linearized system must contain an unstable eigenvalue.  That is, $\bm x_{\text{eq}i}$ cannot be a stable equilibrium. 
	\end{itemize}
	
	We then show that there is no closed orbit in Case (c).  By selecting a continuously differentiable function $g(\bm x)= x_1(t)$, it follows from $x_1(t) \in (0,r_d)$ and $x_2(t) \in (-\pi, 0)$ that
	\begin{align*}
	\frac{\partial(g(\bm x) \dot x_1)}{\partial x_1}+ \frac{\partial(g(\bm x) \dot x_2)}{\partial x_2} =  -c_1 \omega_c x_1(t)\sin x_2(t) >0.
	\end{align*}	
	Our claim follows from the Dulac's Criterion \cite[Chapter 7.2]{Strogatz2018Nonlinear}.  
	
	Since there is no stable equilibrium  or closed orbit of \eqref{eq36} in Case (c), we recall Poincar{\'e}-Bendixson Theorem \cite[Chapter 7.3]{Strogatz2018Nonlinear} to conclude that Case (c) cannot always hold. Particularly,  there is a finite $\delta>0$ such that $d(t_0+\delta)\ge r_d$ with $\phi(t_0+\delta) \in [-\pi/2,0)$, which is {\bf Case (a)}, or $\phi(t_0+\delta)\in [0,\pi]$.

	Combining the above cases,  there exists a finite time instant $t_1$ such that $\phi(t_1)\in [0,\pi]$ for any initial $\phi(t_0) \in (-\pi,0)$. 
\end{proof}

\begin{lemma} \label{lemma2}
	Under the conditions in Proposition \ref{prop2},  the closed-loop system in (\ref{eq5}) is asymptotically stable.  
\end{lemma}
\begin{proof}
	By Lemma \ref{lemma1}, there exists a finite time instant $t_1$ such that $x_2(t) \in [0,\pi]$, $\forall t\ge t_1$.
	
	Consider the Lyapunov function candidate as 
	\begin{align*}
		V\left(\bm x \right) &=  \int_{r_d}^{x_1(t)} \frac{ c_2}{v} \sat\left(\frac{\tau -r_d}{r_d}\right)\text{d}\tau + 
		\int_{r_d}^{x_1(t)}\left( \frac{1}{r_d} -\frac{1}{\tau}\right) \text{d}\tau \\
		&~~~~+ 1- \sin x_2(t).
	\end{align*}
	Taking the time derivative of $V(\bm x)$ along with \eqref{eq5} leads to  
	\begin{align} \label{eq22}
		&\dot V\left(\bm x\right)\nonumber\\
		&=  c_2  \sat\left(\frac{x_1(t) -r_d}{r_d}\right) \cos x_2(t) +\left(\frac{v}{r_d}-\frac{v}{x_1(t)}\right)\cos x_2(t)\nonumber \\
		&~~~~- \left(\omega(t) -\frac{v \sin x_2(t)}{x_1(t)} \right) \cos x_2(t) \nonumber\\
		&= - v\cos x_2(t) 
		\left(  \frac{1}{x_1(t)} -\frac{\sin x_2(t)}{x_1(t)} +  \frac{c_1 }{r_d} \cos x_2(t)  \right ),
	\end{align}
where the last equality is obtained by substituting $\omega(t)$ in \eqref{eq222}.

	 If $x_2(t) \in [0,\pi/2]$, we have that $\cos x_2(t)\ge 0$ and $1-\sin x_2(t)\ge 0$. It follows from \eqref{eq22} that $\dot V(\bm x) \le 0.$  
	
	 If $x_2(t) \in (\pi/2,\pi]$, then $\cos x_2(t)< 0$. To determine the sign of $\dot V(\bm x)$, the following three cases are considered. 
	\begin{itemize}
		\item[(i)] For $x_1(t)\ge r_d$,  it follows from $c_1 > 1$  in (\ref{eqc}) that 
		$
		c_1 /r_d \cdot  \cos x_2(t)< 1/r_d\cdot \cos {x_2(t)} \le {1}/{x_1(t)}\cdot\cos {x_2(t)}.
		$
		Then, it holds that 
			\begin{align*}
				\dot V\left(\bm x\right)<& -v\frac{\cos x_2(t)}{x_1(t)}\left(1+\cos {x_2(t)} -\sin x_2(t)\right)\\
					                          =& -v\frac{\cos x_2(t)}{x_1(t)}\left(1+\sqrt{2}\cos(x_2(t)+\pi/4)\right)\\					                 
					                          \le& ~ 0,
				\end{align*}
			where the last inequality follows from that $3\pi/4< x_2(t)+\pi/4\le 5\pi/4$.
		\item[(ii)] For $v/(c_1 \omega_c)<x_1(t) < r_d$, it holds that
		$
		x_1(t) >  - \left( v(1 - \sin x_2(t)) \right)/\left( c_1 \omega_c\cos x_2(t)\right).
		$ Then
		\begin{align*}
			{1}/{x_1(t)} - {1}/{x_1(t)} \cdot \sin x_2(t) +  {c_1 /r_d}\cdot \cos x_2(t) < 0.
		\end{align*}
		Jointly with \eqref{eq22}, it can be easily verified that $\dot V(\bm x)<0.$
		\item[(iii)] For $0<x_1(t)\le v/(c_1 \omega_c)$, it follows from (\ref{eq36}) that 
		\begin{align*}
			\dot x_2(t) &\le \omega_c +    \frac{c_2 }{ r_d} ( {d(t)-r_d}) + c_1 \omega_c(\cos \phi(t)-\sin \phi(t)) \\ 
			&< \omega_c -c_1 \omega_c +  {c_2 }/{r_d}\cdot  ( {d(t)-r_d})\\
			&<(1 -c_1) \omega_c<0.
		\end{align*}
		This implies that $x_2(t)$ will enter the interval $[0,\pi/2]$ in a finite time. Moreover, $x_2(t) = \pi/2$ leads to that
		\begin{align*}
			\begin{cases}
				\dot x_2(t)~<~ 0, & \text{if}~ x_1(t) \in (0,r_d), \\
				\dot x_2(t)~=~0, & \text{if}~x_1(t) = r_d,\\
				\dot x_2(t)~>~ 0, & \text{if}~ x_1(t) \in (r_d,\infty).
			\end{cases}
		\end{align*}
		Thus, the vehicle states never return to  $0<x_1(t)\le v/c_1 \omega_c$ and $\pi/2<x_2\le\pi$.
		Eventually, $\dot V(\bm x) \le 0$.
	\end{itemize}
	
	However, $\dot V(\bm x)$ is not negative definite, and $\dot V(\bm x) =0$ if $x_2(t) = \pi/2$.  
	Let $\mathcal S=\{\bm x  | \dot V(\bm x) =0  \}$, and suppose that $\widetilde{\bm x}_e $ is an element of $\mathcal S$ except $\bm x_e$. Then 
	\begin{align*}
		\dot  x_2 |_{\bm x =\widetilde{ \bm x}_e}=\omega_c -{v}/{x_1(t)} +  {c_2}\sat \left(  e(t) \right)   \neq 0.
	\end{align*}
	That is, no solution can stay identically in $\mathcal S$ other than the trivial solution $\bm x(t) \equiv \bm x_e$.  Clearly, $V(\bm x)$ is nonnegative, and $V(\bm x)>0$, $\forall \bm x \neq \bm x_e$. By the LaSalle's invariance theorem \cite[Corollary 4.1]{Khalil2002Nonlinear}, $\bm x_e$ is an asymptotically stable equilibrium of the closed-loop system \eqref{eq5} under the controller \eqref{eq222}. \end{proof}

If a closed-loop system is locally exponentially stable near the equilibrium,  this system is robust against perturbations \cite[Chapter 9.2]{Khalil2002Nonlinear}.  

{\em Proof of Proposition \ref{prop2}}: By Lemma \ref{lemma2},  the closed-loop system (\ref{eq36}) has a globally stable equilibrium $\bm x_e$.
	Then, its model near this equilibrium is written as follows 
	\begin{equation} \label{eq23}
		\begin{split}
			\dot x_1(t) &= v \cos x_2(t), \\
			\dot x_2(t)	&=\omega_c + {c_1 \omega_c} \cos x_2(t) + \frac{c_2 }{r_d}(x_1(t)-r_d)      - \frac{v\sin x_2(t)}{x_1(t)}.
		\end{split}
	\end{equation}	
	
	Clearly, the linearized system of \eqref{eq23} at $\bm x_e$ is given by
	\begin{align} \label{eq34}
		\dot {\bm x} (t) = F (\bm x(t) -\bm x_e)
	\end{align}	
	where  $F=\bmatri 0 & -v \\  \frac{1}{r_d}{({c_2 }  + {\omega_c})}  &- {c_1 \omega_c} \ematri$  has two eigenvalues with negative real parts, i.e., $F$ is Hurwitz. Let $\mathcal D=\{ \bm x| V(\bm x)\le b \}$, where $b>0$. If $b$ is sufficiently small, then $d(t)$ is sufficiently close to $r_d$ and $\phi(t)$ is sufficiently close to $\pi/2$. Moreover, the closed-loop system in (\ref{eq23}) is continuously differentiable in $ \mathcal D$.  By  \cite[Corollary 4.3]{Khalil2002Nonlinear}, $\bm x_e $ is an exponentially stable equilibrium for the closed-loop system in (\ref{eq5}). Thus, there exists a finite time instant $t_1$ such that $\bm x(t) \in \mathcal D$ for all $t>t_1$. 
	
	Let $F = Q\Lambda Q^{-1}$ and $\Lambda= \diag( \lambda_1,  \lambda_2 )$ where $\lambda_i$, $i=1,2$ are the eigenvalues of matrix $F$. Then, it holds that
	
	\begin{align*}
		\twon{\bm x(t) -\bm x_e } &= \twon{ Q\exp(\Lambda (t-t_1))Q^{-1} (\bm x(t_1)-\bm x_e)} \\
		&\le C\twon{\bm x(t_1)-\bm x_e} \exp(- \rho (t-t_1)), ~\forall t\ge t_1
	\end{align*}
	where $C=\twon{Q}\twon{Q^{-1}}$,  $\Delta=(c_1 \omega_c)^2 -4 (c_2 \omega_c  + \omega_c^2)$, and 
	\begin{align*}
		\rho=
		\begin{cases}
			(c_1 \omega_c - \sqrt{\Delta})/2, &~\text{if}~ \Delta > 0,\\
			c_1 \omega_c/2, &~\text{if}~ \Delta \le 0.
		\end{cases}
	\end{align*}
	
	Thus, the proof is completed.\qed


\subsection{Proof of Proposition 2} \label{suba2}

\begin{lemma}  \label{lemma_moving}
	Under the conditions in Proposition \ref{prop_moving}, there is a finite time instant $t_1\ge t_0$ such that $\phi(t_1)\in[\arccos((q_1+\bar v_o)/v),~\pi-\arccos(-v_*-\bar v_o-q_1)/v)]$ and $d(t_1)>2r_d$, where $q_1=c_2 r_d /c_1$ and  $0<v_*<v-\bar v_o -q_1$.
\end{lemma}
\begin{proof}
	If $\phi(t_0)\in[\arccos((q_1+\bar v_o)/v),~\pi-\arccos((-v_*-\bar v_o-q_1)/v)]$, the proof is finished. Thus, we only need to analyze the case that $\phi(t_0)$ does not belong to this interval.
	
	When $d(t)\ge2 r_d$,  
	it follows from \eqref{eq63} that
	\begin{align} \label{eqa63}
		\dot\phi(t) =~ & \omega_c+ {c_1}/{r_d}\cdot \left(v \cos\phi(t) - v_1(t)+   q_1 \right) \nonumber \\
		 & -{v\sin \phi(t)}/{d(t)}  + {v_2(t)}/{d(t)}.
	\end{align}
	
	\begin{itemize}	
		\item[(i)] If $\phi(t) \in(-\arccos( (\bar v_o-q_1)/v),~\arccos((q_1+\bar v_o)/v))$, then \eqref{eqa63} leads to that  
		$
		\dot\phi(t) >~ (v-\bar v_o)/{2r_d}+  c_1/r_d \cdot(\bar v_o- v_1(t) )>0.
		$
		\item[(ii)] If $\phi(t) \in (\pi-\arccos((-v_*-\bar v_o-q_1)/v),\pi]\cap (-\pi, -\pi+\arccos((-v_*-\bar v_o-q_1)/v)]$,  then it follows from \eqref{eq63} and the conditions in Proposition \ref{prop_moving} that
		$
		\dot\phi(t) 
		<- c_1/r_d \cdot v_* +\omega_c   + \omega_o<0.
		$
		Moreover, the maximum time for $\phi(t)$ to cross the boundary of $\pi-\arccos((-v_*-\bar v_o-q_1)/v)$ is given as
		\begin{align*}
			T_3 = \frac{2\arccos((-v_*-\bar v_o-q_1)/v)}{c_1/r_d \cdot v_* -\omega_c   - \omega_o}.
		\end{align*} 
		\item[(iii)] If $\phi(t) \in(-\pi+\arccos((-v_*-\bar v_o-q_1)/v),-\arccos( (\bar v_o-q_1)/v))$, it follows from \eqref{eqa63} and Lemma \ref{lemma1} that $\phi(t)$ either reduces to case (i) or case (ii) in a finite time. That is 
		\begin{align*}
			T_4 = \max\left\{\frac{\pi}{c_1/r_d \cdot v_* -\omega_c   - \omega_o}, ~\frac{\pi}{c_2 -c_1 \omega_o}\right\}.
		\end{align*}
		
		
	\end{itemize}

	Thus, there is a finite time instant $t_1\ge t_0$ such that  $\phi(t_1)\in[\arccos( (\bar v_o+q_1)/v),\pi-\arccos((-v_*-\bar v_o-q_1)/v)]$. 
	
	Moreover, case (i) implies that $\dot d(t)>0$ by \eqref{eq19}, and only cases (ii)-(iii) may result in the decrease of $d(t)$. Thus, it holds that $d(t_1)>2r_d$ by $d(t_0)>2r_d + (v+\bar v_o)T_1$ where $T_1=T_3+T_4$.
\end{proof}

\section*{Acknowledgment}
The authors would like to thank Mr. Bo Yang from Tsinghua University for his great support on the experimental tests with the Racecar.

\bibliographystyle{IEEEtran}
\bibliography{bib/mybib}

\begin{thebibliography}{10}
\providecommand{\url}[1]{#1}
\csname url@samestyle\endcsname
\providecommand{\newblock}{\relax}
\providecommand{\bibinfo}[2]{#2}
\providecommand{\BIBentrySTDinterwordspacing}{\spaceskip=0pt\relax}
\providecommand{\BIBentryALTinterwordstretchfactor}{4}
\providecommand{\BIBentryALTinterwordspacing}{\spaceskip=\fontdimen2\font plus
\BIBentryALTinterwordstretchfactor\fontdimen3\font minus
  \fontdimen4\font\relax}
\providecommand{\BIBforeignlanguage}[2]{{%
\expandafter\ifx\csname l@#1\endcsname\relax
\typeout{** WARNING: IEEEtran.bst: No hyphenation pattern has been}%
\typeout{** loaded for the language `#1'. Using the pattern for}%
\typeout{** the default language instead.}%
\else
\language=\csname l@#1\endcsname
\fi
#2}}
\providecommand{\BIBdecl}{\relax}
\BIBdecl

\bibitem{wang2020stationary}
L.~Wang, Y.~Zou, and Z.~Meng, ``Stationary target localization and
  circumnavigation by a non-holonomic differentially driven mobile robot:
  Algorithms and experiments,'' \emph{International Journal of Robust and
  Nonlinear Control}, vol.~31, no.~6, pp. 2061--2081, 2021.

\bibitem{kokolakis2021robust}
N.-M.~T. Kokolakis and N.~T. Koussoulas, ``Robust standoff target tracking with
  finite-time phase separation under unknown wind,'' \emph{Journal of Guidance,
  Control, and Dynamics}, vol.~44, no.~6, pp. 1183--1198, 2021.

\bibitem{olavo2018robust}
J.~L.~G. Olavo, G.~D. Thums, T.~A. Jesus, L.~C. De~Araujo~Pimenta, L.~A.~B.
  Torres, and R.~M. Palhares, ``Robust guidance strategy for target circulation
  by controlled {UAV},'' \emph{IEEE Transactions on Aerospace and Electronic
  Systems}, vol.~54, no.~3, pp. 1415--1431, 2018.

\bibitem{oh2015coordinated}
H.~Oh, S.~Kim, H.-S. Shin, and A.~Tsourdos, ``Coordinated standoff tracking of
  moving target groups using multiple {UAVs},'' \emph{IEEE Transactions on
  Aerospace and Electronic Systems}, vol.~51, no.~2, pp. 1501--1514, 2015.

\bibitem{yoon2013circular}
S.~Yoon, S.~Park, and Y.~Kim, ``Circular motion guidance law for coordinated
  standoff tracking of a moving target,'' \emph{IEEE Transactions on Aerospace
  and Electronic Systems}, vol.~49, no.~4, pp. 2440--2462, 2013.

\bibitem{swartling2014collective}
J.~O. Swartling, I.~Shames, K.~H. Johansson, and D.~V. Dimarogonas,
  ``Collective circumnavigation,'' \emph{Unmanned Systems}, vol.~2, no.~03, pp.
  219--229, 2014.

\bibitem{Lawrence2003Lyapunov}
D.~Lawrence, ``Lyapunov vector fields for {UAV} flock coordination,'' in
  \emph{2nd AIAA Unmanned Unlimited Conference, Workshop, and Exhibit}.\hskip
  1em plus 0.5em minus 0.4em\relax AIAA, 2003, pp. 1--8.

\bibitem{Frew2008Coordinated}
E.~W. Frew, D.~A. Lawrence, and M.~Steve, ``Coordinated standoff tracking of
  moving targets using {L}yapunov guidance vector fields,'' \emph{Journal of
  Guidance Control \& Dynamics}, vol.~31, no.~2, pp. 290--306, 2008.

\bibitem{chen2013uav}
H.~Chen, K.~Chang, and C.~S. Agate, ``{UAV} path planning with
  {tangent-plus-Lyapunov} vector field guidance and obstacle avoidance,''
  \emph{IEEE Transactions on Aerospace and Electronic Systems}, vol.~49, no.~2,
  pp. 840--856, 2013.

\bibitem{Dong2019Flight}
F.~Dong, K.~You, and J.~Zhang, ``Flight control for {UAV} loitering over a
  ground target with unknown maneuver,'' \emph{IEEE Transactions on Control
  Systems Technology}, vol.~28, no.~6, pp. 2461 -- 2473, 2019.

\bibitem{Beard2014Fixed}
R.~W. Beard, J.~Ferrin, and J.~Humpherys, ``Fixed wing {UAV} path following in
  wind with input constraints,'' \emph{IEEE Transactions on Control Systems
  Technology}, vol.~22, no.~6, pp. 2103--2117, 2014.

\bibitem{yao2021singularity}
W.~Yao, H.~G. de~Marina, B.~Lin, and M.~Cao, ``Singularity-free guiding vector
  field for robot navigation,'' \emph{IEEE Transactions on Robotics}, vol.~37,
  no.~4, pp. 1206--1221, 2021.

\bibitem{Oliveira2016Moving}
T.~Oliveira, A.~P. Aguiar, and P.~Encarnação, ``Moving path following for
  unmanned aerial vehicles with applications to single and multiple target
  tracking problems,'' \emph{IEEE Transactions on Robotics}, vol.~32, no.~5,
  pp. 1062--1078, 2016.

\bibitem{Shames2012Circumnavigation}
I.~Shames, S.~Dasgupta, B.~Fidan, and B.~D.~O. Anderson, ``Circumnavigation
  using distance measurements under slow drift,'' \emph{IEEE Transactions on
  Automatic Control}, vol.~57, no.~4, pp. 889--903, 2012.

\bibitem{cao2019relative}
K.~Cao, Z.~Qiu, and L.~Xie, ``Relative docking and formation control via range
  and odometry measurements,'' \emph{IEEE Transactions on Control of Network
  Systems}, vol.~7, no.~2, pp. 912 -- 922, 2019.

\bibitem{Deghat2014Localization}
M.~Deghat, I.~Shames, B.~D.~O. Anderson, and C.~Yu, ``Localization and
  circumnavigation of a slowly moving target using bearing measurements,''
  \emph{IEEE Transactions on Automatic Control}, vol.~59, no.~8, pp.
  2182--2188, 2014.

\bibitem{chai2014consensus}
G.~Chai, C.~Lin, Z.~Lin, and W.~Zhang, ``Consensus-based cooperative source
  localization of multi-agent systems with sampled range measurements,''
  \emph{Unmanned Systems}, vol.~2, no.~03, pp. 231--241, 2014.

\bibitem{Dobrokhodov2008Vision}
V.~N. Dobrokhodov, I.~I. Kaminer, K.~D. Jones, and R.~Ghabcheloo,
  ``Vision-based tracking and motion estimation for moving targets using
  unmanned air vehicles,'' \emph{Journal of Guidance Control \& Dynamics},
  vol.~31, no.~4, pp. 907--917, 2008.

\bibitem{zhang2015Nonlinear}
M.~Zhang and H.~Liu, ``Nonlinear estimation of a maneuvering target with
  bounded acceleration using multiple mobile sensors,'' \emph{IEEE Transactions
  on Aerospace and Electronic Systems}, vol.~51, pp. 1375--1385, 2015.

\bibitem{Cao2015UAV}
Y.~Cao, ``{UAV} circumnavigating an unknown target under a {GPS}-denied
  environment with range-only measurements,'' \emph{Automatica}, vol.~55, pp.
  150--158, 2015.

\bibitem{zhang2020range}
M.~Zhang, Y.~Lin, H.~Hao, and J.~Mei, ``Range-only control for cooperative
  target circumnavigation of unmanned aerial vehicles,'' \emph{Advanced Control
  for Applications: Engineering and Industrial Systems}, vol.~2, no.~4, pp.
  1--13, 2020.

\bibitem{park2016circling}
S.~Park, ``Circling over a target with relative side bearing,'' \emph{Journal
  of Guidance, Control, and Dynamics}, vol.~39, no.~6, pp. 1454--1458, 2016.

\bibitem{Milutinovi2017Coordinate}
D.~Milutinovi\'{c}, D.~Casbeer, Y.~Cao, and D.~Kingston, ``Coordinate frame
  free {D}ubins vehicle circumnavigation using only range-based measurements,''
  \emph{International Journal of Robust \& Nonlinear Control}, vol.~27, no.~16,
  pp. 2937--2960, 2017.

\bibitem{Matveev2011Range}
A.~S. Matveev, H.~Teimoori, and A.~V. Savkin, ``Range-only measurements based
  target following for wheeled mobile robots,'' \emph{Automatica}, vol.~47,
  no.~1, pp. 177--184, 2011.

\bibitem{matveev2017tight}
A.~S. Matveev, A.~A. Semakova, and A.~V. Savkin, ``Tight circumnavigation of
  multiple moving targets based on a new method of tracking environmental
  boundaries,'' \emph{Automatica}, vol.~79, pp. 52--60, 2017.

\bibitem{Anderson2014A}
R.~P. Anderson and D.~Milutinovic, ``A stochastic approach to {Dubins} vehicle
  tracking problems,'' \emph{IEEE Transactions on Automatic Control}, vol.~59,
  no.~10, pp. 2801--2806, 2014.

\bibitem{dong2019ICCA}
F.~Dong, Y.~Hsu, and K.~You, ``Circumnavigation of an unknown target using
  range-only measurements in {GPS}-denied environments,'' in \emph{IEEE
  International Conference on Control and Automation}.\hskip 1em plus 0.5em
  minus 0.4em\relax IEEE, July 2019, pp. 411--416.

\bibitem{Target2018Guler}
S.~Guler and B.~Fidan, ``Target capture and station keeping of fixed speed
  vehicles without self-location information,'' \emph{European Journal of
  Control}, vol.~43, 06 2018.

\bibitem{Lin20163}
J.~Lin, S.~Song, K.~You, and C.~Wu, ``{3-D} velocity regulation for
  nonholonomic source seeking without position measurement,'' \emph{IEEE
  Transactions on Control Systems Technology}, vol.~24, no.~2, pp. 711--718,
  2016.

\bibitem{Moreno2012Strict}
J.~A. Moreno and M.~Osorio, ``Strict {Lyapunov} functions for the
  super-twisting algorithm,'' \emph{IEEE Transactions on Automatic Control},
  vol.~57, no.~4, pp. 1035--1040, 2012.

\bibitem{dong2019Target}
F.~Dong, K.~You, and S.~Song, ``Target encirclement with any smooth pattern
  using only range-based measurements,'' \emph{Automatica}, vol. 116, pp. 1--9,
  2020.

\bibitem{Khalil2002Nonlinear}
H.~K. Khalil, \emph{Nonlinear Systems (3rd Edition)}.\hskip 1em plus 0.5em
  minus 0.4em\relax Upper Saddle River, NJ, USA: Prentice Hall, 2002, ch. 2, 3,
  4, 9.

\bibitem{small}
\BIBentryALTinterwordspacing
J.~Lee, ``Small fixed wing {UAV} simulator,'' Apr. 2016. [Online]. Available:
  \url{https://github.com/magiccjae/ecen674}
\BIBentrySTDinterwordspacing

\bibitem{dong2020PD}
F.~Dong and B.~Yang, ``Experimental validation of the {PD-like} controller for
  circumnavigating a moving target,'' https://youtu.be/7Pxerdout-Q,
  https://youtu.be/x8OLCOLOR0s, 2020.

\bibitem{Beard2012Small}
R.~W. Beard and T.~W. Mclain, \emph{Small Unmanned Aircraft: Theory and
  Practice}.\hskip 1em plus 0.5em minus 0.4em\relax Princeton, NJ, USA:
  Princeton University Press, 2012, ch.~3.

\bibitem{Strogatz2018Nonlinear}
S.~H. Strogatz, \emph{Nonlinear Dynamics and Chaos: with Applications to
  Physics, Biology, Chemistry, and Engineering (2nd Edition)}.\hskip 1em plus
  0.5em minus 0.4em\relax Boca Raton: CRC Press, 2018, ch.~7.

\end{thebibliography}

\begin{IEEEbiography}
	[{\includegraphics[width=1in,height=1.25in,clip,keepaspectratio]{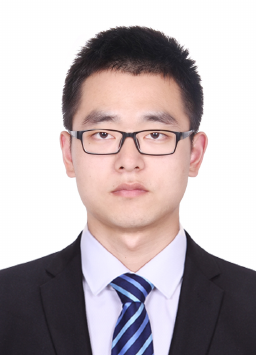}}]
	{Fei Dong} received the B.S. degree from the School of Control Science and Engineering, Shandong University, Jinan, China, in 2014 and the M.S. degree from the School of Instrumentation and Optoelectronic Engineering,  Beihang University, Beijing, China, in 2017. He is currently pursuing the Ph.D. degree at the Department of Automation, Tsinghua University, Beijing, China. His research interests include path planning, motion control, and learning-based control.
\end{IEEEbiography}
\begin{IEEEbiography}
	[{\includegraphics[width=1in,height=1.25in,clip,keepaspectratio]{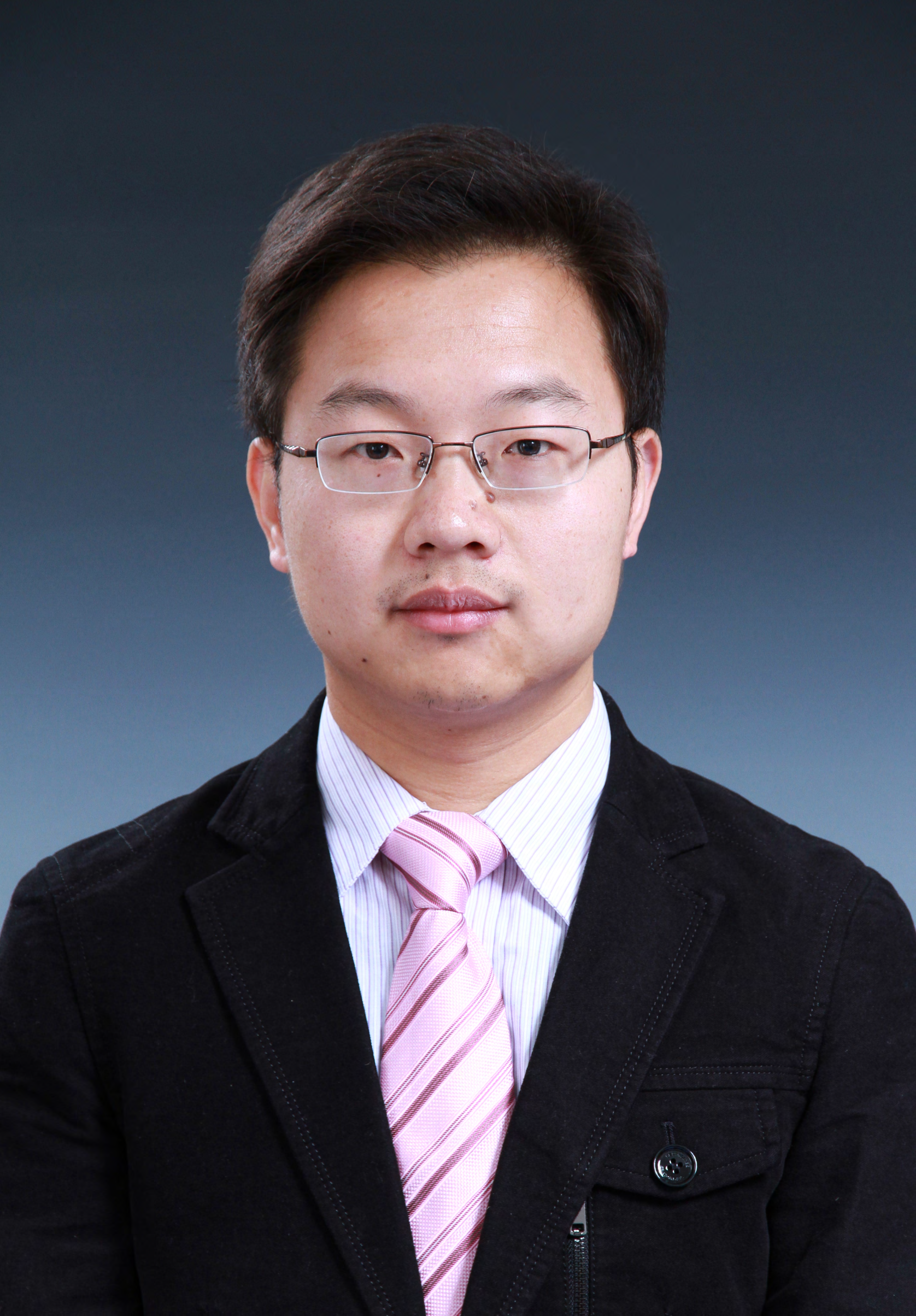}}]
	{Keyou You} (Senior Member, IEEE) received the B.S. degree in Statistical Science from Sun Yat-sen University, Guangzhou, China, in 2007 and the Ph.D. degree in Electrical and Electronic Engineering from Nanyang Technological University (NTU), Singapore, in 2012. After briefly working as a Research Fellow at NTU, he joined Tsinghua University in Beijing, China where he is now a tenured Associate Professor in the Department of Automation. He held visiting positions at Politecnico di Torino, Hong Kong University of Science and Technology, University of Melbourne and etc. His current research interests include networked control systems, distributed optimization and learning, and their applications.
	
	Dr. You received the Guan Zhaozhi award at the 29th Chinese Control Conference in 2010 and the ACA (Asian Control Association) Temasek Young Educator Award in 2019. He received the National Science Fund for Excellent Young Scholars in 2017. He is serving as an Associate Editor for IEEE Transactions on Control of Network Systems, IEEE Transactions on Cybernetics, IEEE Control Systems Letters(L-CSS), Systems $\&$ Control Letters.
\end{IEEEbiography}

\begin{IEEEbiography}
	[{\includegraphics[width=1in,height=1.25in,clip,keepaspectratio]{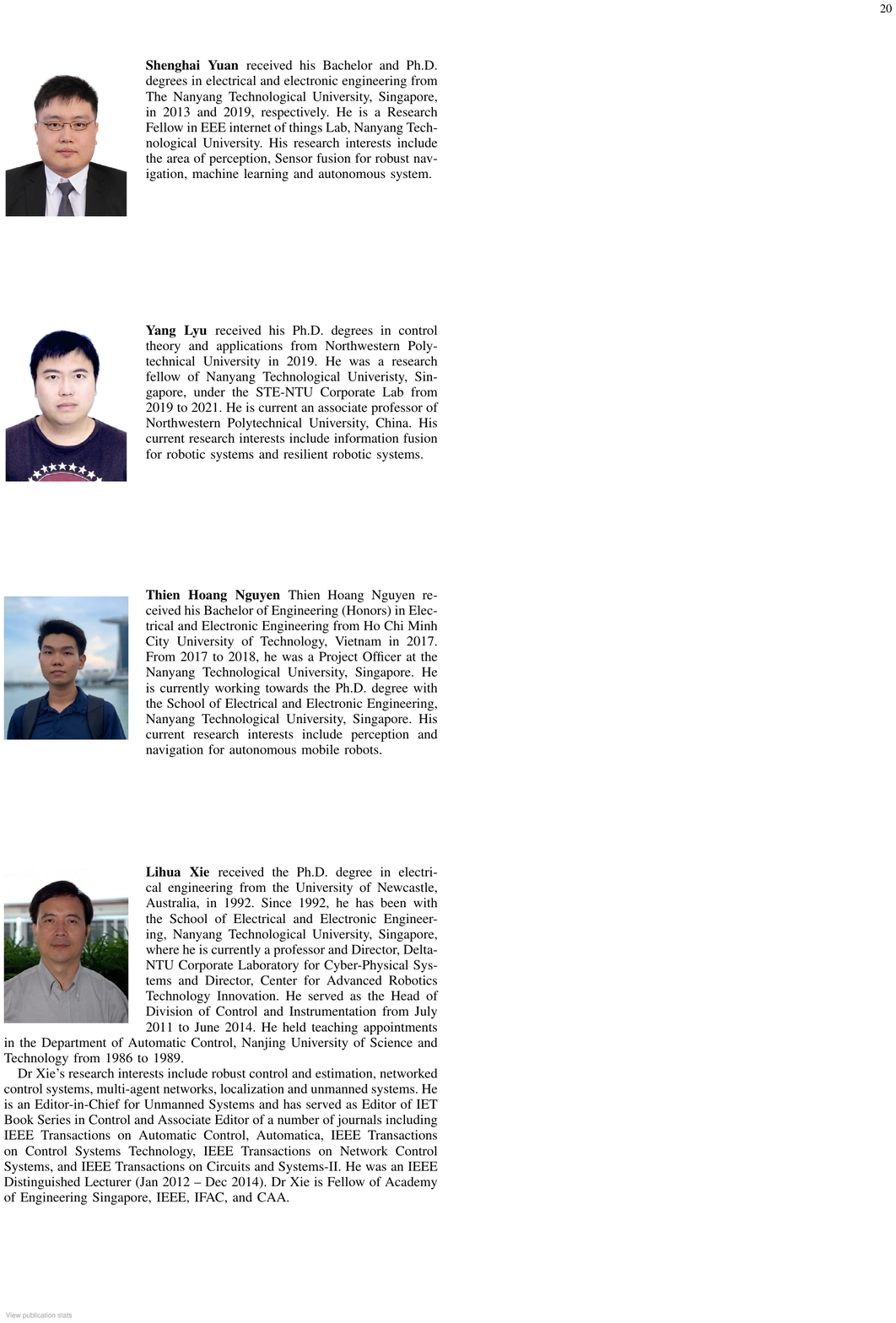}}]
	{Lihua Xie} (Fellow, IEEE)  received the Ph.D. degree in electrical engineering from the University of Newcastle, Australia, in 1992. Since 1992, he has been with the School of Electrical and Electronic Engineering, Nanyang Technological University, Singapore, where he is currently a professor and Director, Delta-NTU Corporate Laboratory for Cyber-Physical Systems and Director, Center for Advanced Robotics	Technology Innovation. He served as the Head of Division of Control and Instrumentation from July 2011 to June 2014. He held teaching appointments
	in the Department of Automatic Control, Nanjing University of Science and Technology from 1986 to 1989.
	
	Dr. Xie's research interests include robust control and estimation, networked control systems, multi-agent networks, localization and unmanned systems. He is an Editor-in-Chief for Unmanned Systems and has served as Editor of IET	Book Series in Control and Associate Editor of a number of journals including IEEE Transactions on Automatic Control, Automatica, IEEE Transactions on Control Systems Technology, IEEE Transactions on Network Control Systems, and IEEE Transactions on Circuits and Systems-II. He was an IEEE
	Distinguished Lecturer (Jan. 2012 - Dec. 2014). Dr Xie is Fellow of Academy of Engineering Singapore, IEEE, IFAC, and CAA. 
\end{IEEEbiography}

\begin{IEEEbiography}
	[{\includegraphics[width=1in,height=1.25in,clip,keepaspectratio]{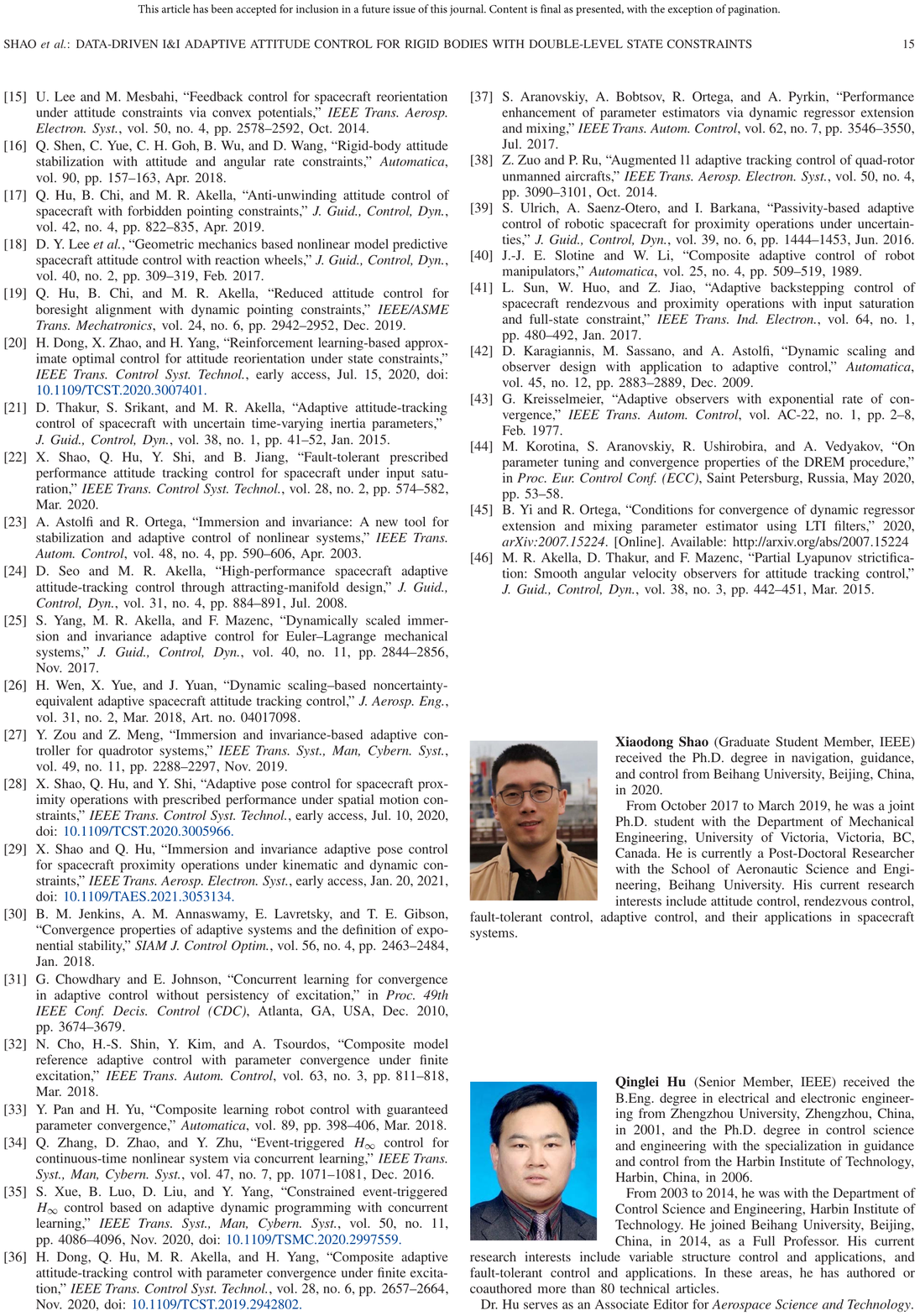}}]
	{Qinglei Hu} (Senior Member, IEEE) received the B. Eng. degree in electrical and electronic engineering from Zhengzhou University, Zhengzhou, China, in 2001, and the Ph.D. degree in control science and engineering with the specialization in guidance
	and control from the Harbin Institute of Technology, Harbin, China, in 2006.
	
	From 2003 to 2014, he was with the Department of Control Science and Engineering, Harbin Institute of Technology. He joined Beihang University, Beijing, China, in 2014, as a Full Professor. His current research interests include variable structure control and applications, and fault-tolerant control and applications. In these areas, he has authored or coauthored more than 80 technical articles.
	
	Dr. Hu is serving as an Associate Editor for Aerospace Science and Technology. 
\end{IEEEbiography}

\end{document}